%% file: fractals_algorithmica.tex
\documentclass{article}[12]
\pdfoutput=1
\usepackage{amsmath}
\usepackage{amsfonts}
\usepackage{amssymb}
\usepackage{amsthm}
\usepackage{ifpdf}
\usepackage{wrapfig}
\usepackage{fullpage}
\usepackage{cite}
\usepackage{graphicx}
\usepackage{subfigure}

\ifpdf

  \usepackage[pdftex]{epsfig}

\else

    \usepackage[dvips]{epsfig}
    \newcommand{\href}[2]{#2}

\fi

\theoremstyle{definition}
\newtheorem{theorem}{Theorem}[section]
\newtheorem{lemma}[theorem]{Lemma}

\newtheorem{definition}[theorem]{Definition}

\begin{document}
\title{Strict Self-Assembly of Fractals using Multiple Hands}
\author{
Cameron T. Chalk\footnotemark[1]\\
\and
Dominic A. Fernandez\footnotemark[1]\\
\and
Alejandro Huerta\footnotemark[1]\\
\and
Mario A. Maldonado\footnotemark[1]\\
\and
Robert T. Schweller\footnotemark[1]\\
\and
Leslie Sweet\footnotemark[1]
}
\footnotetext[1]{Department of Computer Science, University of Texas - Pan American,
\{ctchalk,dafernandez1,ahuerta5,mamaldonado6,
rtschweller,lgsweet\}@utpa.edu
This author's research was supported in part by National Science Foundation Grant CCF-1117672.}

\date{}
\maketitle


\begin{abstract}
In this paper we consider the problem of the strict self-assembly of infinite fractals within tile self-assembly.  In particular, we provide tile assembly algorithms for the assembly of the discrete Sierpinski triangle and the discrete Sierpinski carpet within a class of models we term the \emph{$h$-handed assembly model} ($h$-HAM), which generalizes the 2-HAM to allow up to $h$ assemblies to combine in a single assembly step.  Despite substantial consideration, no purely growth self-assembly model has yet been shown to strictly assemble an infinite fractal without significant modification to the fractal shape.  In this paper we not only achieve this, but in the case of the Sierpinski carpet are able to achieve it within the 2-HAM, one of the most well studied tile assembly models in the literature.  Our specific results are as follows.  We provide a $6$-HAM construction for the Sierpinski triangle that works at scale factor 1, 30 tile types, and assembles the fractal in a \emph{near perfect} way in which all intermediate assemblies are finite-sized iterations of the recursive fractal. We further assemble the Sierpinski triangle within the $3$-HAM at scale factor 3 and 990 tile types.  For the Sierpinski carpet, we present a $2$-HAM result that works at scale factor 3 and uses 1,216 tile types.  We further include analysis showing that the aTAM is incapable of strictly assembling the non-tree Sierpinski triangle considered in this paper, and that multiple hands are needed for the near-perfect assembly of the Sierpinski triangle and the Sierpinski carpet.
\end{abstract}

\pagebreak

\input{introduction}
\input{definitions}
\input{infinite_fractals}

\input{impossibility}

\input{future_work}

\bibliographystyle{amsplain}
\bibliography{tam}

\end{document}

%% file: introduction.tex
\section{Introduction}

\subsection{Self-Assembly.} Self-assembly is the process by which systems of simple objects autonomously organize themselves through
local interactions into larger, more complex objects. Self-assembly processes are abundant in nature and
serve as the basis for biological growth and replication. Understanding how to design molecular self-assembly systems that build complex, algorithmically specified shapes and patterns promises to be of fundamental importance for the future of nanotechnology.  Some of the most interesting algorithmically specified shapes are discrete fractals such as the Sierpinski triangle.  In this paper, we take on the classical problem of designing self-assembling systems for the strict assembly of discrete fractal patterns.  We consider this problem in the context of \emph{multiple handed} tile self-assembly, a natural generalization of two-handed tile self-assembly.

\subsection{Tile Assembly Model: Seeds, Hands, and more Hands.}  In the \emph{tile assembly model} system components are 4-sided Wang tiles~\cite{Wang63} which randomly agitate within the plane and stick together when bumping tile edges share glues with strong enough affinity. Tile attachment is random and based on tiles appearing in random places in the plane. This basic framework has been studied extensively and has also been developed into multiple different variants.  In this paper we consider three such variants: the \emph{abstract tile assembly model} (aTAM)~\cite{Winf98}, the two-handed tile assembly model (2-HAM)~\cite{doty:PIHSA, SS2013FEC, oneTile2014, SRTSARE, DDFIRSS07, fu2012SAGT,2HABTO,AGKS05g,CheDot12,SFTSAFT,DPR2013ICALP}, and a new generalization of the the 2-HAM termed the $h$-HAM.  The aTAM is the original and most widely studied model and is characterized by a \emph{seed} tile from which growth occurs by the attachment of singleton tiles given sufficient bonding strength.  In contrast, the 2-HAM has no seed and allows any two already built and potentially large assemblies to translate together and stick given enough bonding strength.  In the $h$-HAM presented in this paper, we generalize the $2$-HAM to allow up to a given number $h$ pre-built assemblies to come together to form a new producible assembly.  Although the $h$-HAM is seemingly implausible in experimentation settings, it is interesting to explore the properties of such a natural, more parallelized generalization of the $2$-HAM.  Further, the $h$-HAM may prove useful for designing $2$-HAM systems which are robust to multi-handed attachments as a sort of error-prevention mechanism for hierarchical self-assembly.

\subsection{Strict Self-Assembly of Discrete Fractals.}  We focus our consideration of the power of multiple handed self-assembly on a considerable problem in algorithmic self-assembly theory:  the strict self-assembly of discrete fractals.  The self-assembly of discrete fractals and infinite patterns has been a benchmark problem in self-assembly that has yielded both theoretical~\cite{WinBek03,padilla2013ASP,KS2013SAR,SSADST,ASAOTST,SADS,jSADSSF,RNAPods,KSL2009SAD,RTFDA} and experimental work~\cite{RothSier}.  The interest in self-assembled fractals stems from their geometric aperiodicity, which forces the development of new fundamental self-assembly design techniques, along with their frequent occurrence in natural biological systems.  However, much of the previous work on self-assembled fractals centers on the \emph{weak} assembly (described in~\cite{SADS, jSADSSF,RNAPods}) of fractal patterns in which the assembled shape is permitted to place tiles of a special color in locations that are not part of the fractal shape, or permits some level of substantial \emph{warping} of the fractal shape in which tiles are placed outside of the fractal shape in some limited fashion~\cite{ASAOTST, SSADST, KS2013SAR,jSADSSF}, or adds substantial power to the assembly model such as the ability to detach previously built subassemblies~\cite{padilla2013ASP}.  A challenging problem is the \emph{strict} assembly of fractals in a purely growth (non-detaching) model in which the shape of the assembly is exactly the shape of the fractal.  The difficulty stems from the intricate, thin features that characterize infinite fractals.  Pure growth models must decide how to build a shape through a single mechanism, the placement of tiles, which inherently uses up geometric space.  When that space is drastically constrained, typical approaches for self-assembly fail, and in many cases the failure is unavoidable~\cite{jSADSSF, SPFDNSS, STFDNSS}.  In fact, no pure growth model of self-assembly has yet been shown to strictly self-assemble a fractal pattern such as the Sierpinski triangle or Sierpinski carpet without substantial modification to the fractal shape that goes beyond a simple scale factor increase.

\subsection{Our Results.}  We focus on two discrete fractals, a non-tree version of the discrete Sierpinski triangle, and the standard Sierpinski carpet.  Note that the non-tree version of the discrete Sierpinski triangle considered in this paper is a different discretization than previous studies of fractal assembly.  We present two systems that strictly self-assemble the non-tree Sierpinski triangle.  The first system works within the $6$-handed model, uses a surprisingly small number of tile types (30), achieves optimal scale factor of 1, and achieves a desirable property we term \emph{near-perfect} assembly in which the intermediate assemblies of the system are almost exactly the finite-sized iterations that define the infinite fractal shape.  The second system works within the $3$-handed assembly model and achieves scale factor 3.  Finally, we present our most intricate result which is the strict self-assembly of the Sierpinski carpet at scale factor 3 within the $2$-HAM model.  We accompany these positive results with some impossibility results which, along with the positive results, show provable increases in power with multiple handed assembly. In particular, we show that no aTAM system self-assembles our non-tree Sierpinski triangle,  and we show that the non-tree Sierpinski triangle and the Sierpinski carpet cannot be near-perfectly assembled with fewer than $3$ and $8$ hands respectively.  Our results are summarized in Table~\ref{resultsTables}.

\subsection{Outline.}  The remainder of this paper is organized as follows.  In Section~\ref{sec:definitions} we define the aTAM, the 2-HAM, and the $h$-HAM tile assembly models, as well as the non-tree Sierpinski triangle and Sierpinski carpet fractals.  In Sections ~\ref{sec:triangle} and ~\ref{sec:carpet} we present our positive results for the strict self-assembly of the non-tree Sierpinski triangle and the Sierpinski carpet respectively.  In Section~\ref{sec:impossibility}, we present our impossibility results. Finally, in Section ~\ref{sec:futureWorks}, we present some directions for future work that arise from the study of fractal self-assembly and the $h$-HAM.

\begin{center}
\begin{table}
\caption{Our results for the non-tree Sierpinski triangle and the Sierpinski carpet are summarized here.}
\label{resultsTables}
\begin{tabular}{lccccc}
\hline\noalign{\smallskip}
\multicolumn{6}{c}{Infinite Fractals}  \\
\noalign{\smallskip}\hline\noalign{\smallskip}
Fractal & Model & Scale & Tiles & Near Perfect? & Theorem  \\
\hline\noalign{\smallskip}
Sierpinski Triangle & 3-HAM & 3 & 990 & No & \ref{S_3HAM}  \\
Sierpinski Triangle & 6-HAM & 1 & 30 & Yes & \ref{S_6HAM}  \\
Sierpinski Carpet & 2-HAM & 3 & 1,216 & Impossible for $h<8$ & \ref{C_2HAM}, \ref{thm:nearPerfectImpossible}  \\
\hline\noalign{\smallskip}
Sierpinski Triangle & aTAM & $c$ & \multicolumn{2}{c}{Impossible} & \ref{thm:aTAMImpossible}  \\
\hline\noalign{\smallskip}
\end{tabular}
\end{table}
\end{center} 

%% file: definitions.tex
\section{Definitions and Model}
\label{sec:definitions}

\subsection{Tiles and Assemblies}
\paragraph{Tiles.}  Consider some alphabet of glue types $\Pi$.  A tile is a finite edge polygon with some finite subset of border points each assigned some glue type from $\Pi$.  Further, each glue type $g \in \Pi$ has some non-negative integer strength $str(g)$.  Finally, each tile may be assigned a finite length string \emph{label}, e.g., ``black",``white",``0", or ``1".  In this paper we consider a special class of tiles that are unit squares of the same orientation with at most one glue type per edge, with each glue being placed exactly in the center of the tile's edge. Further, for simplicity, we assume each tile center is located at an integer coordinate in $\mathbb{Z}^2$.

\paragraph{Assemblies.}  An assembly is a set of tiles at unique coordinates in $\mathbb{Z}^2$.  For a given assembly $\Upsilon$, define the \emph{bond graph} $G_\Upsilon$ to be the weighted graph in which each element of $\Upsilon$ is a vertex, and each edge weight between tiles is the sum of the strengths of the overlapping, matching glue points of the two tiles.  An assembly $C$ is said to be \emph{$\tau$-stable} for positive integer $\tau$ if the bond graph $G_C$ has min-cut at least $\tau$, and \emph{$\tau$-unstable} otherwise. Note that if the set of border points of all tiles in an assembly is not a connected set, then the assembly cannot be $\tau$-stable (for positive $\tau$).  Note that only overlapping glues that are the same type contribute a non-zero weight, whereas overlapping, non-equal glues always contribute zero weight to the bond graph.  Relaxing this restriction has been considered as well~\cite{AGKS05g}.  Additionally, with the square tiles considered in this paper, stable assemblies will necessarily consist of tiles stacked face to face, forming a subset of the 2D grid denoted by the tiled locations in the assembly. We refer to any subset of an assembly $A$ as a subassembly of $A$. For an assembly $A$ and integer vector $\vec{v} = (v_1, v_2)$, let $A_{\vec{v}}$ denote the assembly obtained by translating each tile in $A$ by vector $\vec{v}$.  A \emph{shape} is any subset of $\mathbb{Z}^2$, and the shape of an assembly $A$ is defined to be the shape consisting of the set of coordinate locations for the centers of each tile in $A$.


\paragraph{Combinable Assemblies.}
Informally, a set of $h$ assemblies is said to be $\tau$-combinable if the assemblies can be translated together to form a $\tau$-stable assembly $C$.  Formally, a set of stable assemblies $\{A_1, \ldots, A_h\}$ is said to be $\tau$-combinable into an assembly $C$ if there exists a set of assemblies $\{A'_1,\ldots, A'_h\}$ such that each $A'_i$ is a translation of $A_i$, and $C= \bigcup_{i=1}^h A'_i$ is a $\tau$-stable assembly.

\subsection{Abstract Tile Assembly Model (aTAM)}
A \emph{tile system} in the Abstract Tile Assembly Model (aTAM) is an ordered triple $\Gamma =(T,\sigma,\tau)$ where $T$ is a set of tiles called the \emph{tile set}, $\sigma$ is an assembly over $T$ called the \emph{seed}, and $\tau$ is a positive integer called the \emph{temperature}.  As notation we denote the set of all tiles that are a translation of some tile in $T$ as set $T^*$.  Assembly in the aTAM proceeds by growing from assembly $\sigma$ by any sequence of single tile attachments from $T$ so long as each tile attachment connects with strength at least $\tau$.  Formally, we define what can be built in this fashion as the set of producible assemblies:

\begin{definition}[aTAM Producibility]
    For a given aTAM system $\Gamma=(T,\sigma,\tau)$, the set of \emph{producible assemblies} for system $\Gamma$, $\texttt{PROD}_\Gamma$, is defined recursively:
    \begin{itemize}
        \item (Base) $\sigma \in \texttt{PROD}_\Gamma$
        \item (Recursion) For any $A\in \texttt{PROD}_\Gamma$ and $b \in T^*$ such that $C=A\bigcup \{b\}$ is $\tau$-stable, then $C \in \texttt{PROD}_\Gamma$.
    \end{itemize}
\end{definition}
As additional notation, we say $A \rightarrow^\Gamma_1 B$ if $A$ may grow into $B$ through a single tile attachment, and we say $A \rightarrow^\Gamma B$ if $A$ can grow into $B$ through 0 or more tile attachments.  For a shape $S$, we say a system $\Gamma$ \emph{uniquely assembles} $S$ if for all $A \in \texttt{PROD}_\Gamma$ there exists a $B\in \texttt{PROD}_\Gamma$ of shape $S$ such that $A \rightarrow_\Gamma B$.  An \emph{assembly sequence} is a way to denote a particular possible sequence of growth for a given system and is formally defined to be any sequence (finite or infinite) of assemblies $\langle A_0, A_1, A_2, \ldots \rangle$ such that $A_0 = \sigma$, and $A_i \rightarrow^\Gamma_1 A_{i+1}$ for each $i$.  For a more detailed description of the aTAM and related concepts and metrics, see~\cite{2HABTO}.

\subsection{Multiple Handed Tile Assembly Model ($h$-HAM)}
The multiple handed tile assembly model generalizes the two-handed assembly model (2-HAM) to allow for potentially more than two hands.  Formally, a multiple handed tile assembly system is an ordered triple $(T,\tau,h)$ where $T$ is a set of tiles called the \emph{tile set}, $\tau$ is a positive integer called the \emph{temperature}, and $h$ is a positive integer called the \emph{number of hands}.  In this paper we sometimes refer to a system $(T,\tau,h)$ as an $h$-HAM system $(T,\tau)$.  Assembly in an $h$-HAM system $(T,\tau)$ proceeds by repeatedly combining up to $h$ assemblies at a time to form new $\tau$-stable assemblies.  Formally, we define what can be built in this fashion as the set of producible assemblies:

\begin{definition}[$h$-HAM Producibility]
    For a given $h$-HAM system $\Gamma =(T,\tau)$, the set of \emph{producible assemblies} for system $\Gamma$, $\texttt{PROD}_\Gamma$, is defined recursively:
    \begin{itemize}
        \item (Base) $T^* \subseteq \texttt{PROD}_\Gamma$
        \item (Recursion) For any set of $r$ assemblies, $0\leq r \leq h$, if a producible set of assemblies $A=\{A_1,\ldots, A_r\} \subseteq \texttt{PROD}_\Gamma$ is $\tau$-combinable into $B$, then $B \in \texttt{PROD}_\Gamma$.
    \end{itemize}
\end{definition}

Typically, we require that an initial assembly set $T$ consist only of stable assemblies.  In this paper, we further restrict $T$ to consist of singleton tile assemblies and thus refer to $T$ as the \emph{tile set} of the system, and we refer to $|T|$ as the \emph{tile complexity} of $\Gamma$.  In the case of $h=2$, we have the standard $2$-HAM tile assembly model in which assembly proceeds by repeatedly combining any pair of combinable assemblies.  We encourage the reader to see~\cite{2HABTO} for a number of 2-HAM related concepts and metrics not explicitly discussed here in this compact definition.  One important concept we use in this paper is the assembly of infinite shapes.  In this paper, we use the \emph{finite assembly}~\cite{2HABTO} of infinite shapes as our definition of what the assembly of infinite shapes means for an $h$-HAM system.  A system $\Gamma$ is said to finitely assemble an infinite shape $X \subset \mathbb{Z}^2$ if every finite producible assembly of $\Gamma$ has a possible way of growing into an assembly that places tiles exactly on those points in $X$ (appropriately translated if necessary).  For the remainder of this paper we will use the term ``self-assembles" rather than ``finitely self-assembles" as there is only one such concept under consideration in this paper.  

\subsection{Sierpinski Shapes}
In this section we define the non-tree discrete Sierpinski triangle and the discrete Sierpinski carpet.  The Sierpinski triangle we consider in this paper is not the shape considered in earlier work in tile self-assembly. The main distinction is that our Sierpinski triangle is not a \emph{tree} shape.  This is important since our positive results rely fundamentally on the non-tree shape of the fractal.  In the case of the tree version of the Sierpinski triangle, \cite{2HABTO} have shown that the shape cannot be strictly assembled at scale factor 1 in either the aTAM or the 2-HAM, and this proof applies more generally to the $h$-HAM.  However, it remains unproven (but we conjecture) that the tree Sierpinski triangle cannot be strictly assembled at any scale within the $h$-HAM.
 Throughout the remainder of the paper, we refer to the non-tree Sierpinski triangle as the Sierpinski triangle.

\begin{definition} \label{def:s_triangle}
	(Discrete Sierpinski Triangle).
	Let $S_{i+1} = S_i \cup (S_i + 2^i V)$  where V = \{(1,2),(-1,2)\} , and $S_0$ = \{(0,0),(-1,0),(0,1),(-1,1)\}.
	The discrete Sierpinski triangle can be defined as the set $S_\infty = \bigcup_{i=0}^\infty S_i$. (Figure \ref{triangle}).
	$S^c_\infty$ and $S^c_i$ are factor $c$ scalings of $S_\infty$ and $S_i$.
\end{definition}

\begin{definition} \label{def:s_carpet}
	(Discrete Sierpinski Carpet).
	Let $C_{i+1} = C_i \cup (C_i + 3^i V)$, where $V$ = \{(1,0), (0,1), (0,2), (2,0), (1,2), (2,1), (2,2)\}, and $C_0 = (0,0)$.
	The discrete Sierpinski carpet will then be the set $C_\infty = \bigcup_{i=0}^\infty C_i$. (Figure \ref{carpet})
	$C^c_\infty$ and $C^c_i$ are factor $c$ scalings of $C_\infty$ and $C_i$.
\end{definition}

 \begin{figure}
\begin{center}
	\begin{minipage}{0.45\textwidth}
		\includegraphics[width=\textwidth]{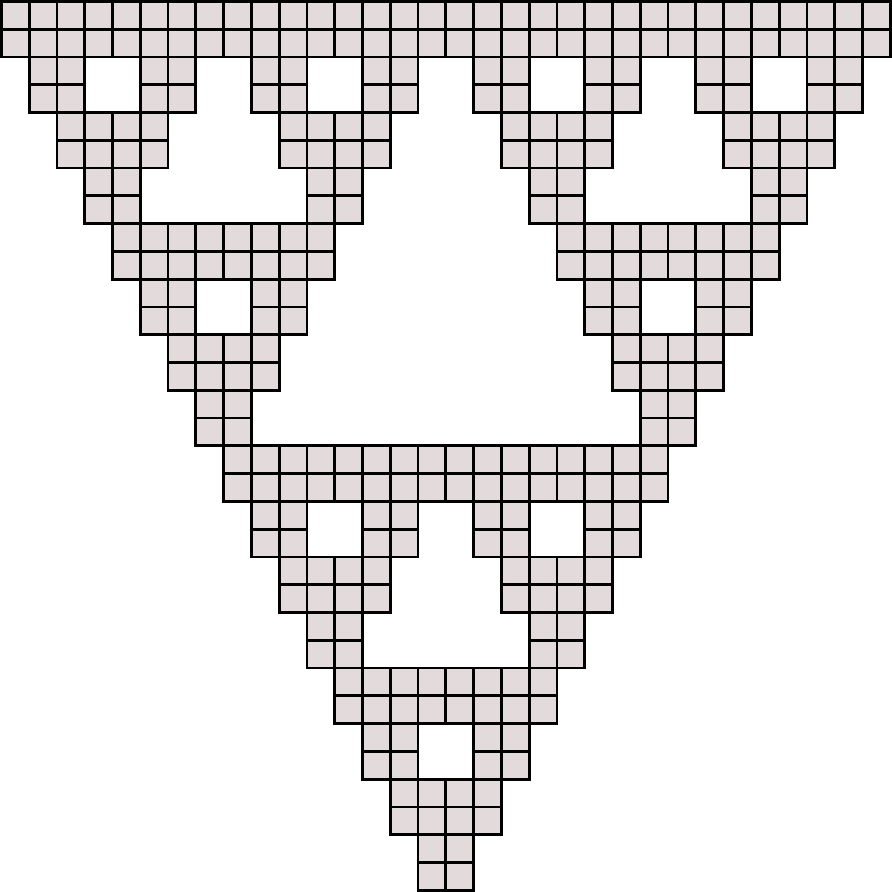}
		
		\caption{Scale 1 Sierpinski triangle. ($S_i^1$)}
        \label{triangle}
		\end{minipage}\hfill
	\begin{minipage}{0.45\textwidth}
		\includegraphics[width=\textwidth]{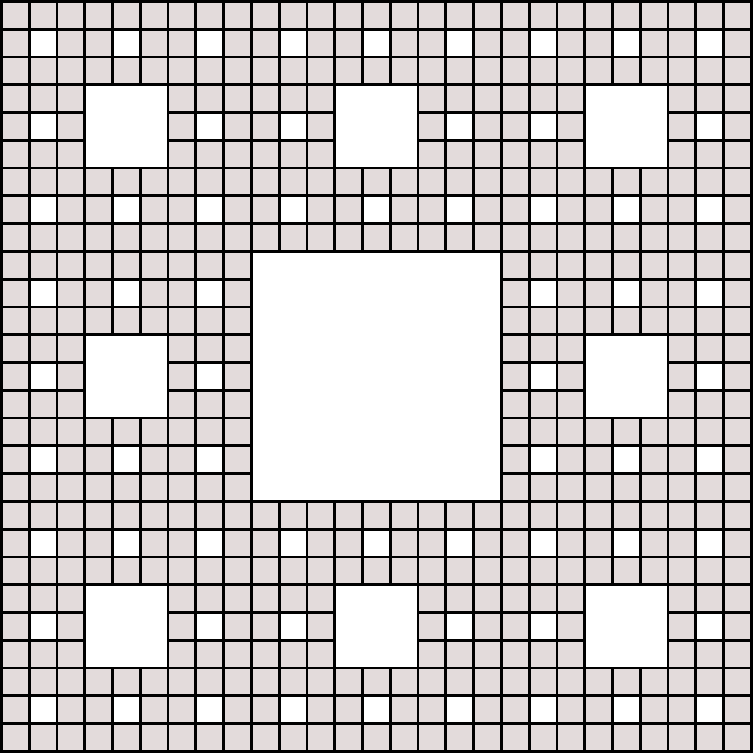}
		\caption{Scale 1 Sierpinski carpet. ($C_i^1$)}

		\label{carpet}
	\end{minipage}
\end{center}
\end{figure}

%% file: infinite_fractals.tex
\section{Strict Self-Assembly of the Discrete Sierpinski Triangle}
\label{sec:triangle}
In this section, we present two constructions for the strict assembly of the non-tree discrete Sierpinski triangle.  Note that the assembly of the standard discrete Sierpinski triangle was proven to be impossible in ~\cite{2HABTO}.  Our first construction is a 6-HAM system that strictly assembles the non-tree Sierpinski triangle at optimal scale 1 and 30 distinct tiles.  Further, this construction achieves \emph{near perfect assembly} (Def. \ref{nearPerfectAssembly}).  Our second construction uses only 3 hands and works at scale 3. We conjecture that 3 is the minimum number of hands required to assemble the non-tree Sierpinski triangle.

\input{S_6HAM}

\input{S_3HAM_v2}

\section{Strict Self-Assembly of the Discrete Sierpinski Carpet}\label{sec:carpet}

\*Weak self-assembly of the discrete Sierpinski carpet has been achieved previously \cite{KSL2009SAD,KS2013SAR}. However, no purely-growth model of tile assembly has been shown to strictly self-assemble the fractal at any scale factor. In this section, we show a construction that strictly self-assembles the discrete Sierpinski carpet in the well-known 2-HAM.

\input{C_2HAM} 

%% file: S_6HAM.tex
\subsection{6-Handed Assembly of the Sierpinski Triangle}

In this section, we provide a 6-HAM system that assembles the discrete Sierpinski triangle at scale 1, using 30 distinct tiles.  The system also assembles the Sierpinski triangle in an idealized manner that we call \emph{near perfect assembly} (Def. \ref{nearPerfectAssembly}).  When we say a system near perfectly assembles an infinite fractal, we mean that the system will only produce assemblies that differ from some finite stage (an index in sequence $P$) of the fractal by a constant number of points.  In other words, if we take any producible assembly in the system, we will find one stage of the fractal that is almost equal to the shape of the chosen assembly with only a constant number of points missing.  The inverse will also be true; if we pick any stage of the fractal, we will find one or more assemblies that are almost equal to that stage.

\begin{definition}(Near Perfect Assembly).
\label{nearPerfectAssembly}
An infinite fractal pattern $ P=\bigcup_{i=n}^\infty P_i$, where sets $P_{i+1}$ are obtained by a function $f(P_i)$, is said to be \emph{near perfectly assembled} by an $h$-HAM system $\Gamma = (T,\tau)$ if the system assembles $P$, and if:
\begin{enumerate}
	\item There exists a $c$ such that for all $i$, there exists an assembly $A \in \texttt{PROD}_\Gamma$ of shape $S_A$ such that $|P_i\setminus S_A| \leq c$. In other words, for every $P_i$, there exists at least one producible assembly $A \in \texttt{PROD}_\Gamma$ whose shape is a subset of $P_i$ smaller by at most a constant number of points independent from $i$.
	\item For any producible assembly $A \in \texttt{PROD}_\Gamma$, the shape of some translation of $A$ is a subset of some $P_i$, smaller by at most a constant number of points independent from $i$.
\end{enumerate}

\end{definition}

\begin{theorem}
\label{S_6HAM}
There exists a 6-HAM system that near perfectly assembles the discrete Sierpinski triangle at scale factor 1, and has a tile complexity of 30.
\end{theorem}

\begin{proof}
The proof is by construction.  We first show how condition 1 is met by presenting the tiles in the tileset, and describing how they will combine to produce assemblies with a shape that is 6 points away from an $S_i$, for all $S_i$.  We then show that condition 2 is also satisfied by doing induction on the size (number of tiles) of the producible assemblies, and proving that the shape of all producible assemblies of size greater than 30 must be a subset of a stage of the Sierpinski triangle, smaller by only 6 points.

Consider 6-HAM system $\Gamma = (T,\tau)$ with tile types $T$ shown in Figure \ref{S_tileset}, and $\tau = 2$.
The tile set $T$ (made up of $30$ tiles) is designed to meet both conditions required for near perfect assembly (Def. \ref{nearPerfectAssembly}) of the discrete Sierpinski triangle.  The near perfect assembly feature of the system is achieved by taking advantage of the six hands and the temperature to prevent any other wrong assemblies from combining.  Each of the large producible assemblies is made up of six pieces, which must be assembled before they combine into the next stage ($S_i$) of the Sierpinski triangle.

\begin{figure}
\begin{center}
	\subfigure[The ``base'' assemblies $b1, b2$, and $b3$.] {
		\includegraphics[width=.40\textwidth]{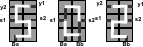}
		\label{fig:tileSetBase}
	} \qquad
	\subfigure[The ``helper'' assemblies. We refer to them as the $\alpha, \gamma$, and $\beta$ assemblies.] {
		\includegraphics[width=.30\textwidth]{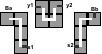}
		\label{fig:tileSetHelper}
	}
\caption{Tile set $T$ used to near perfectly assemble a Sierpinski triangle. White lines show a necessary sequence in which the pieces must be constructed in order to achieve near perfect assembly.}
\label{S_tileset}
\end{center}
\end{figure}

\paragraph{Condition 1}
To satisfy condition 1 of near perfect assembly, all stages ($S_i$) of the discrete Sierpinski triangle need to have at least one assembly with a shape that is almost, by a constant number of points, equal to that stage.
Consider the first stage $S_0$ of the triangle, whose shape is defined to be a 2x2 square with its leftmost bottom point at (-1,0).
An assembly of either $b1$ or $b3$ with the bottommost four tiles will form the squared shape exactly, with no points of difference (Figure \ref{S_missingTiles}).
The next stage, $S_1$, consists of three copies of the $S_0$ shape, one staying in the original's position, one translated up and to the left by 2 and 1 points respectively, and another one translated up and to the right by the same amount.
The shape of $b1, b2$, or $b3$ is the closest to $S_1$, and it is smaller than this stage by 6 tiles (it is missing the two leftmost, bottommost, and rightmost tiles) as shown in Figure \ref{S_missingTiles}.
$S_2$ follows the same pattern, consisting of three copies of $S_1$ translated by the specified amount.
The assembly with the biggest subset of $S_2$ will be composed of the three base assemblies, glued together by the set of $\alpha, \beta$, and $\gamma$ helper assemblies shown in Figure \ref{fig:tileSetHelper}.
These six assemblies will need to form a``loop" (Figure \ref{S_missingTiles}) joining them all together or the assembly will not be $\tau$-stable.

\begin{figure}
\begin{center}
	\includegraphics[width=.70\textwidth]{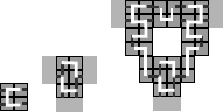}
	\caption{Assemblies used to satisfy condition 1 of near perfect assembly for $S_0, S_1$, and $S_2$.}
	\label{S_missingTiles}
\end{center}
\end{figure}

The resulting assembly's shape is a subset of $S_2$, smaller by 6 points.
Condition 1 for all of the following stages will be met in the same way: three copies of the assembly used to satisfy the previous stage, all joined together by a new set of $\alpha, \beta$, and $\gamma$ assemblies.
This process is shown in Figure \ref{S_stagesFloat} for a small part of the system, but it will be repeated for any $S_i$.

\begin{figure}
\begin{center}
	\includegraphics[width=.90\textwidth]{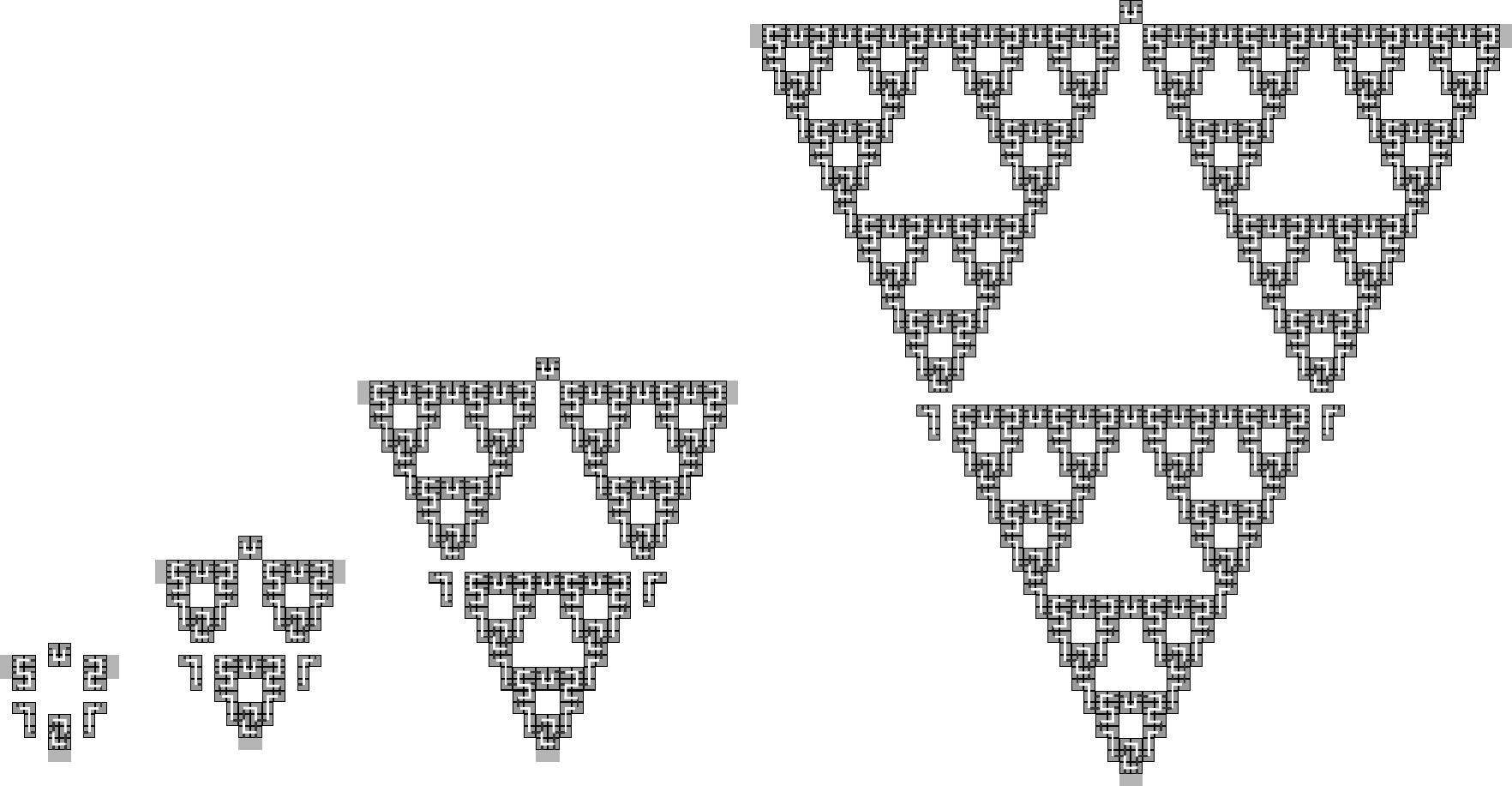}
	\caption{Stages of growth for 6-HAM system $\Gamma$. The floating pieces are separated to easily distinguish between the six different assemblies that combined to form a new iteration from the previous one.}
	\label{S_stagesFloat}
\end{center}
\end{figure}

Note that these specific assemblies are smaller than an $S_i$ for all $S_i$ by 6 tiles, and therefore $\Gamma$ satisfies condition 1 of near perfect assembly of the Sierpinski triangle.

\paragraph{Condition 2}
To see how $\Gamma$ satisfies condition 2, we first define \emph{near-triangle assemblies}.  Near-triangle assemblies are a class of assemblies whose shape is a subset of some stage of the Sierpinski triangle ($S_i$) with the same 6 corner points missing (the two leftmost, rightmost and bottommost points). Near-triangle assemblies must also have the same set of exposed glues at the same relative points.  One example of a near-triangle assembly is the third assembly in Figure \ref{S_missingTiles} that is close to stage $S_2$.  The shape of this assembly is missing the six corner points. Further, tiles at points adjacent to these missing points expose glues that allow only helper tiles to attach.  We use notation $A_i$ to denote the near-triangle assembly with shape close to $S_i$.  Note that all near-triangle assemblies with shape close to $S_i$ are of size $ 4 \times 3^i - 6$.

To show system $\Gamma$ meets condition 2 of near perfect assembly we prove by induction, on the size of the assemblies, that all producible assemblies of size greater than 30 must be near-triangle assemblies.  Note that assemblies smaller than 30 are base or helper assemblies, which can only combine to form near-triangle assembly $A_2$.
Consider the near-triangle assembly shown in Figure \ref{S_missingTiles} (third assembly), of size $s = 30$ as the base case. Since the base case ($A_2$) is a near-triangle assembly we assume all producible assemblies of size $s > 30$, up to assembly $A_n$ of size $ 4\times 3^n - 6$, are near-triangle assemblies of some $S_i$ for $i \leq n$.  Assembly $A_n$ is a near-triangle assembly, which means the six corner points in $S_n$ are missing from $A_n$, and it also exposes the same glues as all other near-triangle assemblies.

Only helper tiles can attach to the exposed glues on near-triangle assemblies, and helper tiles attach only to near-triangle assemblies or base pieces, so the assembly of size $4 \times 3^n-5$ (or any bigger producible assembly) has to follow the same process as the one shown in Figure \ref{S_stagesFloat}, where 3 equal-sized near-triangle assemblies must be combined with the 3 helper assemblies, or it is not producible.
The process must be followed because if the three near-triangle assemblies are not the same size, the distance between at least two assemblies is not filled and the loop described before is not closed, (as shown in Figure \ref{S_mismatchedSize}) resulting in an assembly that is not $\tau$-stable.  Any assembly of size $4\times 3^n-5$ is not producible since the same type of attachments must happen.  In fact, we can be sure $A_n$ can only grow into $A_{n+1}$ which is a near-triangle assembly of size $4\times 3^{n+1} - 6$.
 $A_n$ cannot grow into an assembly that is not $A_{n+1}$, so $\Gamma$ meets condition 2 of near perfect assembly.

\begin{figure}
\begin{center}
	\includegraphics[width=.50\textwidth]{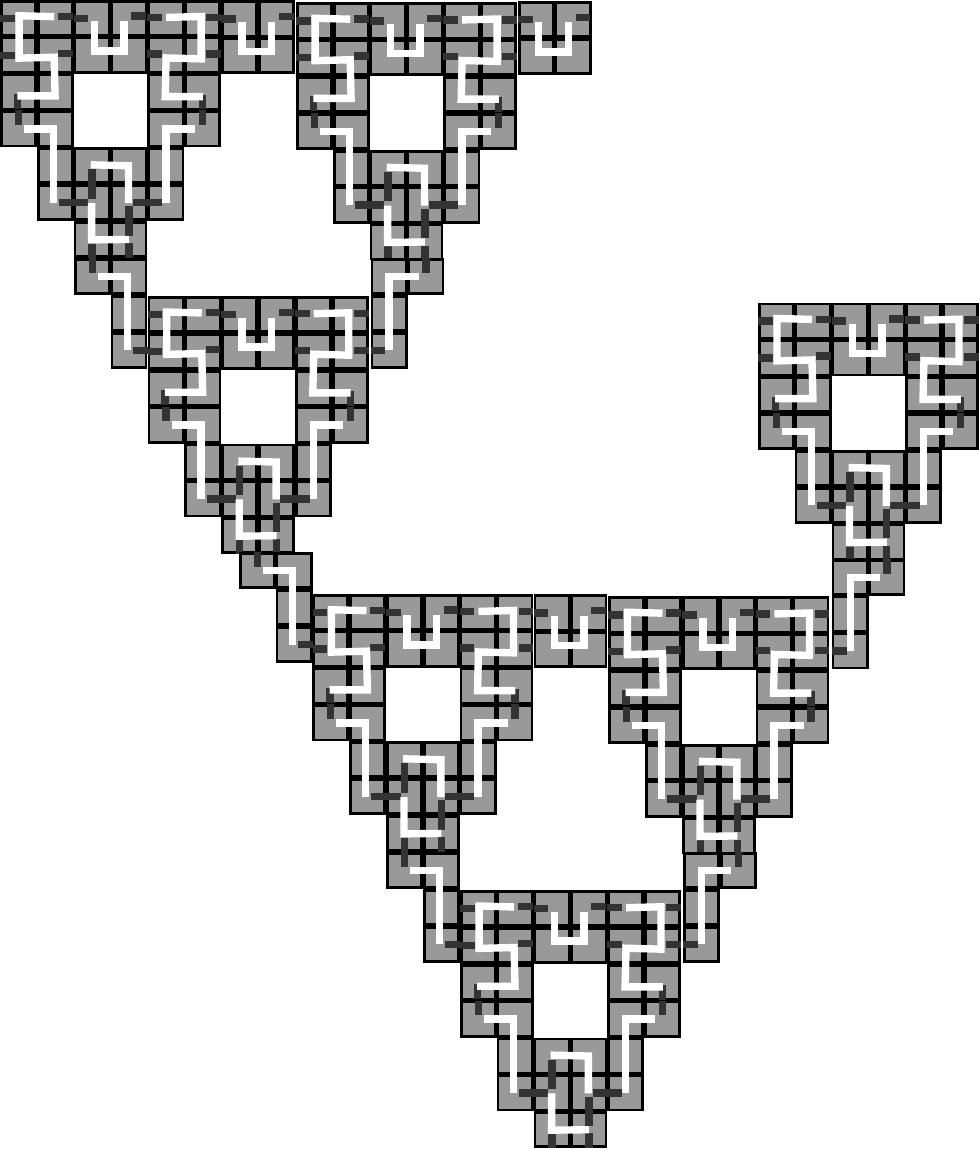}
	\caption{Near-triangle assemblies of different size will never be stable.}
	\label{S_mismatchedSize}
\end{center}
\end{figure}

To Summarize: The shape of any producible assembly smaller than size 30 is smaller than an $S_i$ by at most 9 points (partial assembly of $\alpha$ and $\beta$) and the shape of anything larger than 30 is smaller than an $S_i$ by 6 points.  These are the only producible assemblies, so both conditions are met, and since we have shown and explained how the system assembles the Sierpinski triangle, we conclude that $\Gamma$ near perfectly assembles the discrete Sierpinski triangle.
\end{proof} 

%% file: S_3HAM_v2.tex
\subsection{3-Handed Assembly of the Sierpinski Triangle}

In this section, we provide a 3-HAM system that strictly self-assembles the discrete Sierpinski triangle at scale factor 3.  Note that the system does not near perfectly assemble the Sierpinski triangle, since it relies on many ``filler'' assemblies that increase in number as other assemblies grow.

\begin{theorem}
\label{S_3HAM}
There exists a 3-HAM system $\Gamma = (T, \tau)$ that self-assembles $S^3_{\infty}$ with $|T| = 990$.
\end{theorem}

\begin{proof}
\*
\paragraph{Overview.}  An intuitive overview of our Sierpinski triangle construction follows. A 3-HAM system $\Gamma = (T, 2)$ with a tileset of size 990 is used to self-assemble the discrete Sierpinski triangle. 9 triangles that are nearly of size $S_1^3$ are assembled. These assemblies differ from the shape of $S_1^3$ in that they are missing tiles along two edges to allow for tiles to be placed along edges in later stages of the assembly process.  An example is shown in Figure \ref{yellowBase}. These 9 triangles will attach in triplets to create 3 unique triangles that are nearly of size $S_2^3$. These near-$S_2^3$ assemblies are designed to leave a perimeter such that their edges can be tiled. This process is shown in Figure \ref{summary1}.

\begin{figure}
\begin{center}
	\includegraphics[width=\textwidth]{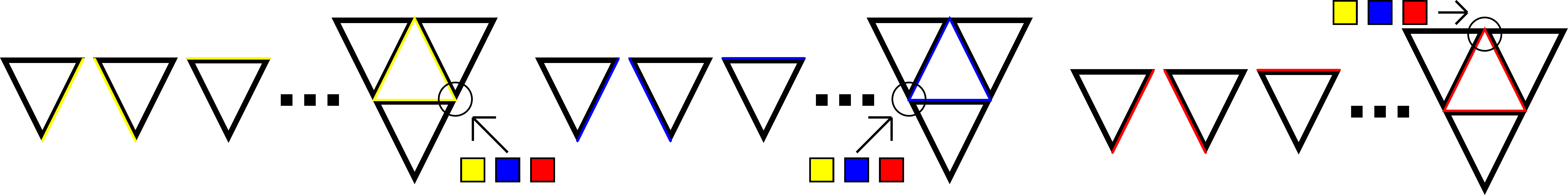}
	\caption{3 sets of 3 unique $S_1^3$ assemblies attach to form 3 unique near-$S_2^3$ assemblies. Keystone attachment happens at the circled region.}
	\label{summary1}
\end{center}
\end{figure}

Once a near-$S_2^3$ assembly is complete, it allows for the attachment of an assembly referred to as a \emph{keystone}. A keystone initiates a growth of tiles along an edge of a near-$S_i^3$ assembly. Once the edge is tiled, the assembly will attach to 2 other near-$S_i^3$ assemblies, forming a near-$S_{i+1}^3$ assembly. The edges of the near-$S_{i+1}^3$ assembly are untiled, and one of the edges exposes glues that allow another attachment of a keystone. One of 3 different keystones will nondeterministically attach. Once the keystone attaches, the keystone binding location is unavailable, therefore only one keystone type can attach to a near-$S_i^3$ assembly. The nondeterministic attachment of a keystone to a near-$S_i^3$ assembly determines on which edge tiles will be placed on a near-$S_{i+1}^3$ assembly composed of 3 near-$S_i^3$ assemblies that have attached that keystone. The 3 near-$S_i^3$ assemblies will only attach to form a near-$S_{i+1}^3$ assembly if they have attached the same keystone. This process is shown in Figure \ref{summary2}. The tiles that fill in edges of a near-$S_i^3$ assembly are of a constant number and are designed to fill in arbitrarily long edges, allowing for a constant-sized tileset.

\begin{figure}
\begin{center}
	\includegraphics[width=.75\textwidth]{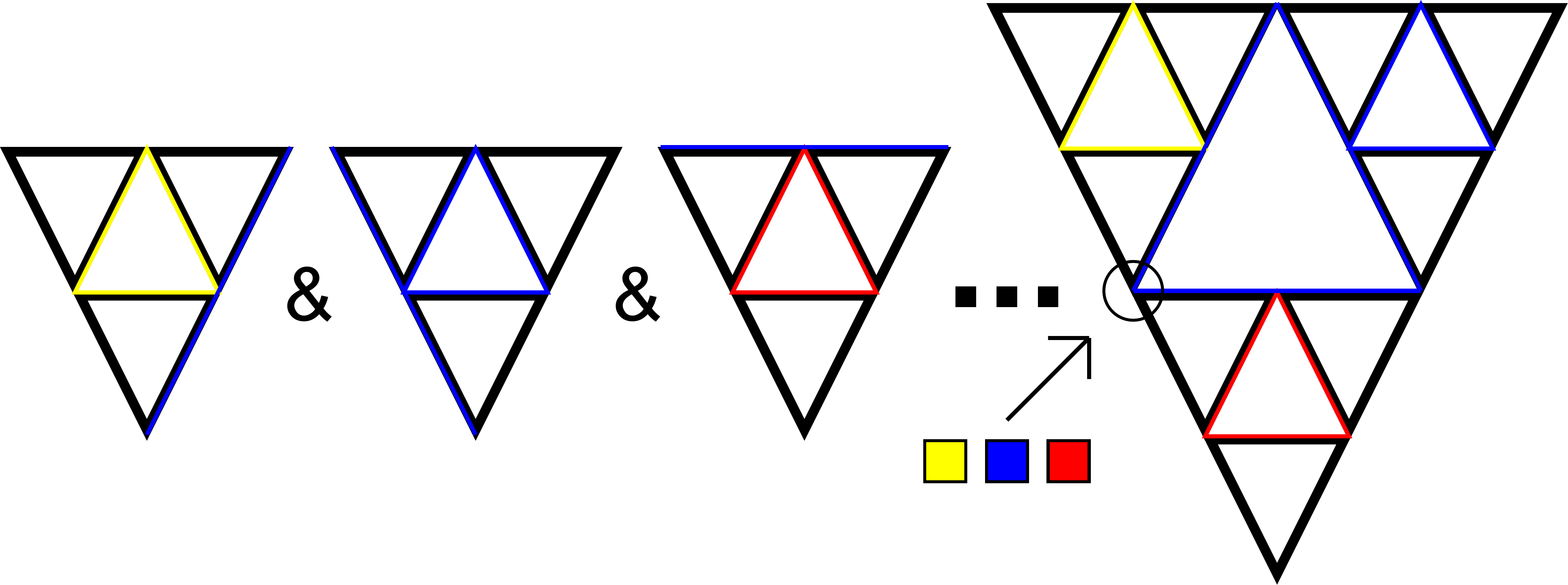}
	\caption{3 near-$S_2^3$ assemblies, having attached the same blue colored keystone, join to form a near-$S_3^3$ assembly. Note that the region that a keystone attachment occurs is dependent on the color of the inner perimeter. The nondeterministic attachment of the keystone determines if the triangle assembled becomes the right, left, or bottom triangle in the next iteration of the assembly.}
	\label{summary2}
\end{center}
\end{figure}

\begin{figure}
\begin{center}
	\includegraphics[width=.85\textwidth]{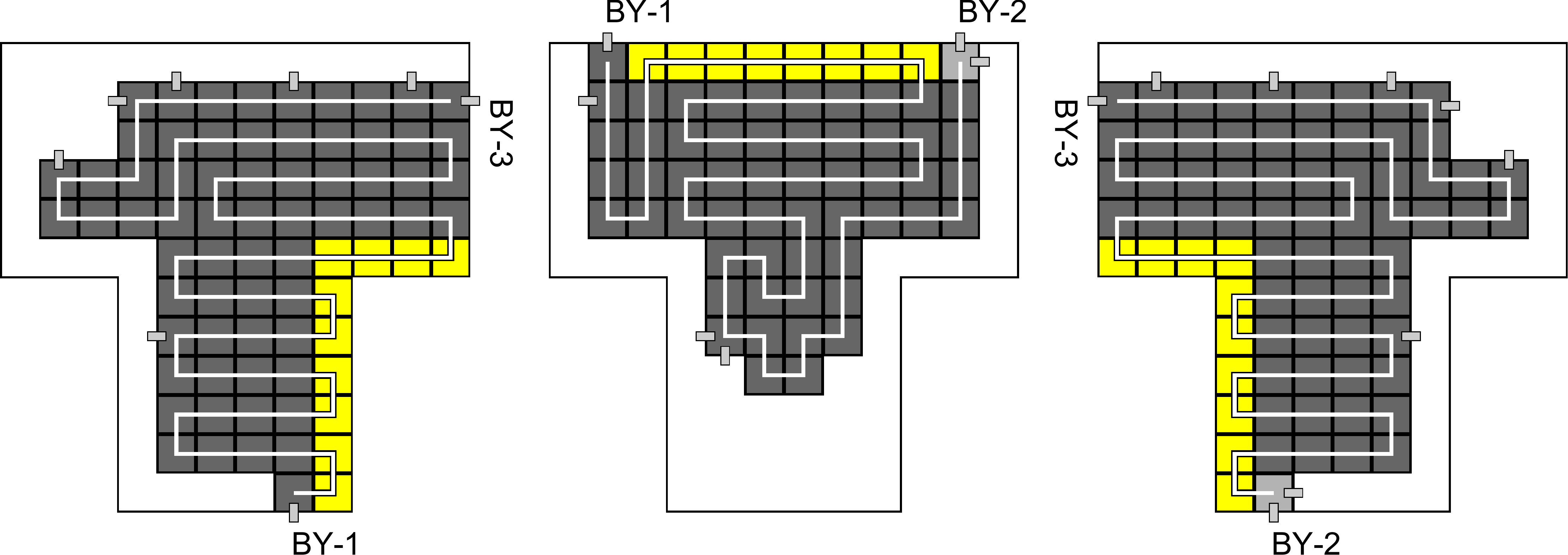}
	\caption{Yellow combinable shapes of $S_1^3$. Once combined, they will become the left triangle in the next iteration. There are corresponding combinable shapes of size $S^3_1$ for the right and bottom triangles.}
	\label{yellowBase}
\end{center}
\end{figure}

\paragraph{Base Shapes.} A \emph{base shape of $S^3_i$} is an assembly with shape $S^3_i$ that is missing tiles along the outer edges and the corners of the triangle as shown in Figure \ref{baseShapes}. A \emph{combinable shape of $S^3_i$} is an assembly with shape $S^3_i$ that is missing tiles along two outer edges of the triangle and has corner tiles where these edges meet as shown in Figure \ref{combinableShapes}.

\begin{figure}
\begin{center}
	\subfigure[Two examples of a base shape.] {
		\includegraphics[width=.45\textwidth]{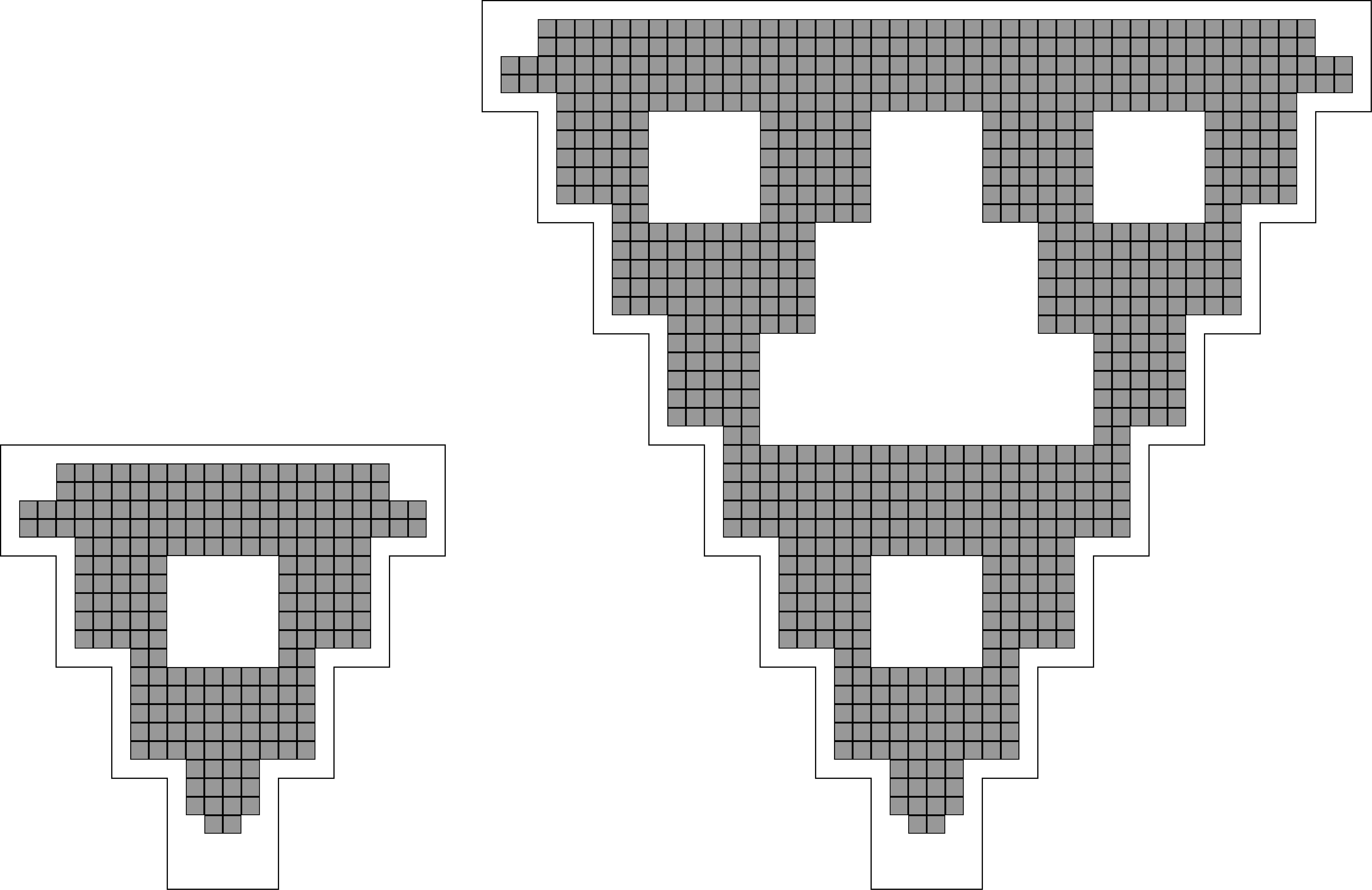}
		\label{baseShapes}
	}
	\subfigure[Two examples of a combinable shape.] {
		\includegraphics[width=.45\textwidth]{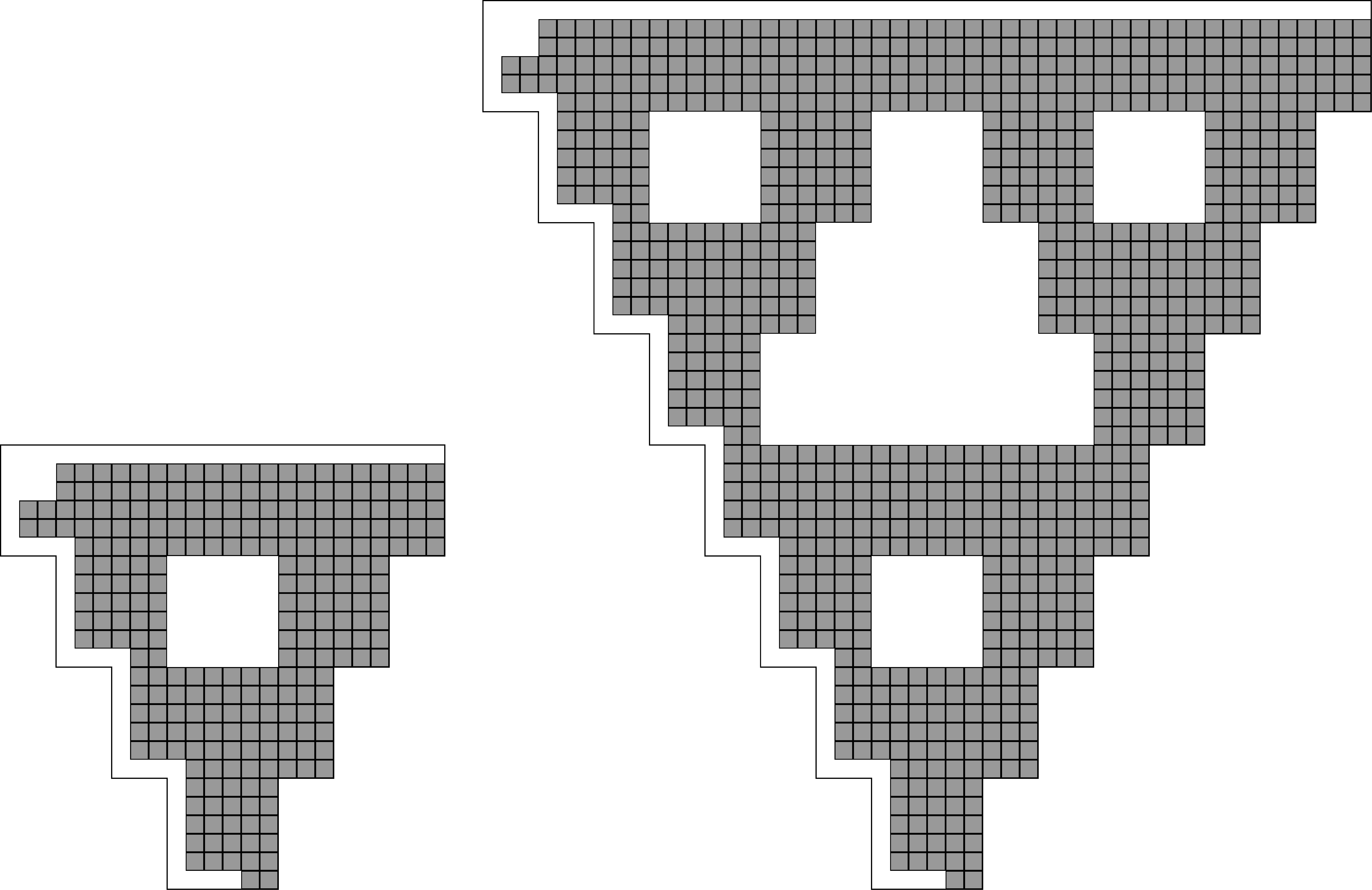}
		\label{combinableShapes}
	}
\caption{Examples of the two types of shapes defined. }
\label{shapes}
\end{center}
\end{figure}

\paragraph{Formation of Base Pieces.} The assembly process in the system begins by assembling the assemblies shown in Figure \ref{yellowBase}, which are combinable shapes of $S^3_1$. There are three sets in total of these combinable shapes of $S^3_1$, two of which are not pictured here. The white path traces the assembly process of these assemblies, where tiles connected by the path are connected by unique strength-2 glues.  This pattern ensures that each combinable shape of $S^3_1$ is completely assembled before either can assemble into the base shape of $S^3_2$. Once these combinable shapes of $S^3_1$ have assembled completely, they can attach to their two corresponding combinable shapes of $S^3_1$.  Each combinable shape of $S^3_1$ has two strength-1 glues that are used to construct a base shape of $S^3_2$ using cooperative binding. Each strength-1 glue attaches to one triangle each, forming a $\tau$-stable assembly, as shown in Figure \ref{4baseProcess_1}. The other two sets of combinable shapes of $S^3_1$ not shown here exhibit the same behavior. These three different base shapes of $S^3_2$ will be described as either yellow, blue, or red. We show the yellow combinable shapes of $S^3_1$ that make up base shape of $S^3_2$ in Figure \ref{yellowBase}.

\paragraph{Keystone Attachments.} \emph{Keystone tiles} are assemblies that cooperatively attach to their corresponding exposed glues, in a nondeterministic manner. For each base shape of $S^3_i$, one of three keystones will attach, followed by the filling of an edge of the assembly. This method of assembly is used to produce three different combinable shapes of $S^3_i$. The assembled base shapes of $S^3_2$, described in the previous paragraph, have two exposed glues on a single edge which are sites of keystone tile attachment. Each assembled base shape of $S^3_2$ has three keystone tiles that can attach at their exposed edges. We will describe these keystone tiles by the colors yellow, blue, and red.  These three keystone tiles are each split into two sets, so that they must be assembled using 3 hands(two for the keystone tiles and one for the base shape of $S^3_2$) to be properly assembled, as shown in Figure \ref{4baseProcess_2}. Each of the three different assembled base shapes of $S^3_2$ expose keystone glues on differing edges; a yellow base shape will attach their keystone tiles on the right edge, a blue base shape on the left edge, and a red base shape on the top edge.

\begin{figure}
\begin{center}
	\subfigure[Three combinable shapes form a base shape.] {
		\includegraphics[width=.30\textwidth]{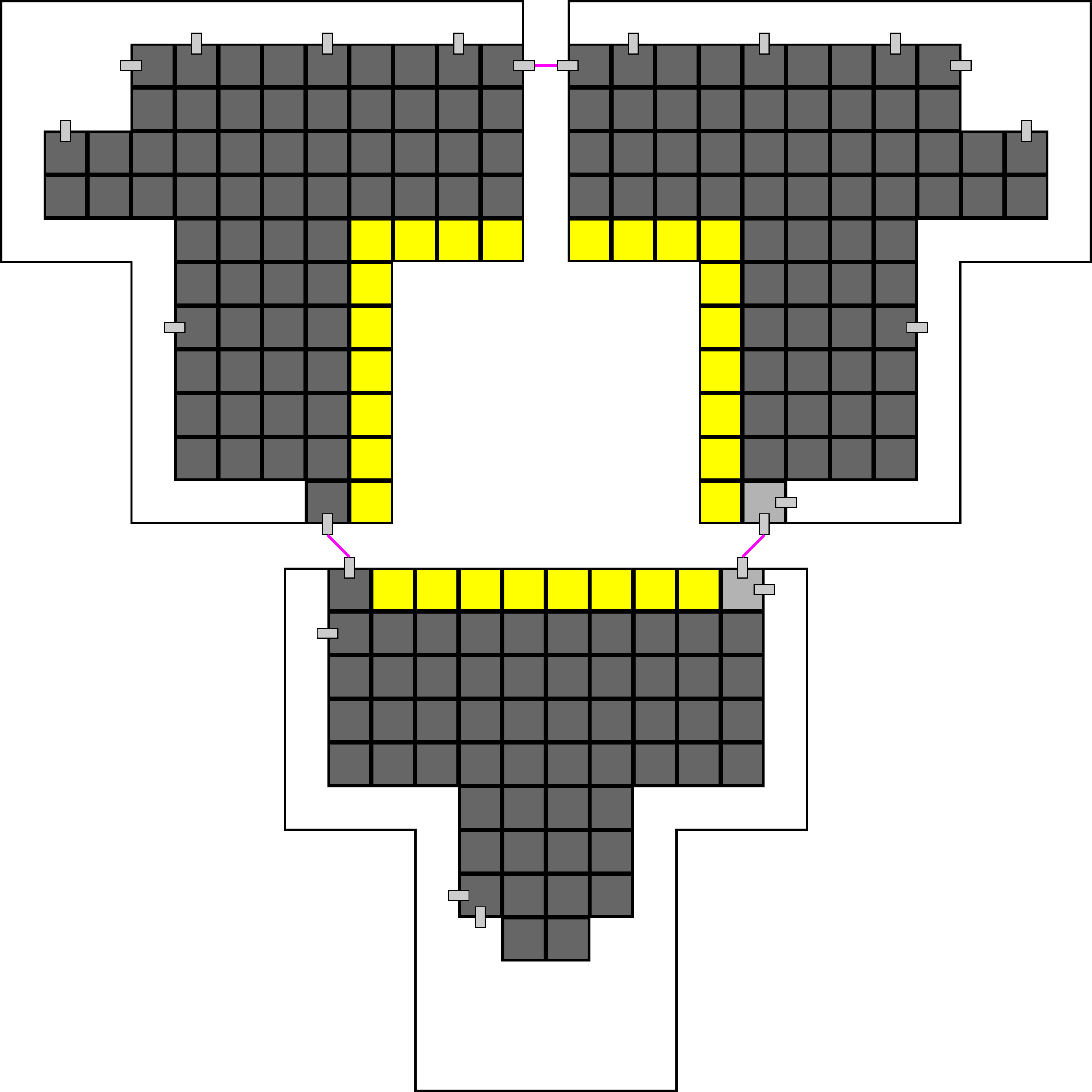}
		\label{4baseProcess_1}
	}
	\subfigure[Keystone assemblies attach to the base shape.] {
		\includegraphics[width=.30\textwidth]{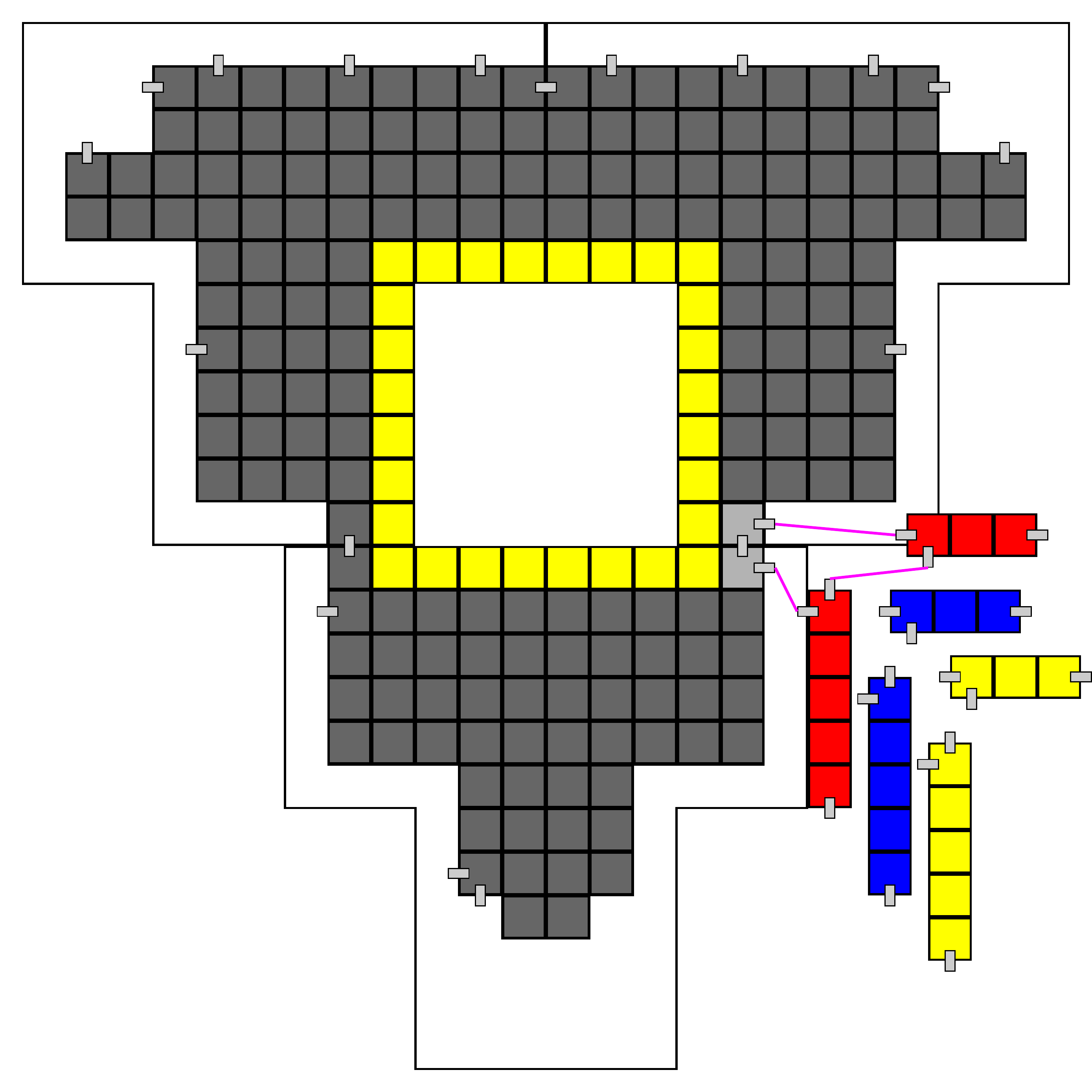}
		\label{4baseProcess_2}
	}
	\subfigure[Right fillers placed.] {
		\includegraphics[width=.30\textwidth]{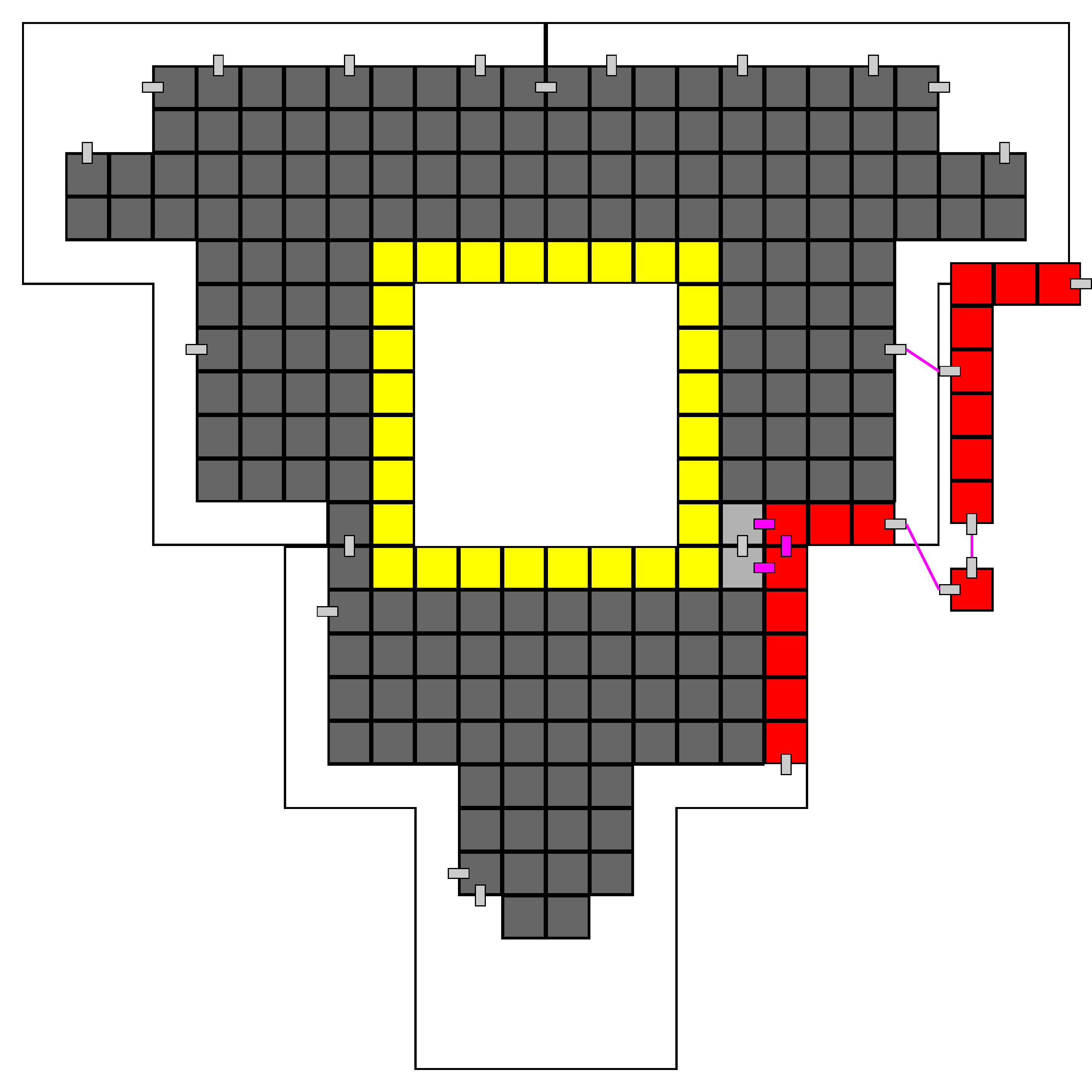}
		\label{4baseProcess_3}
	}
	\subfigure[Right fillers placed.] {
		\includegraphics[width=.30\textwidth]{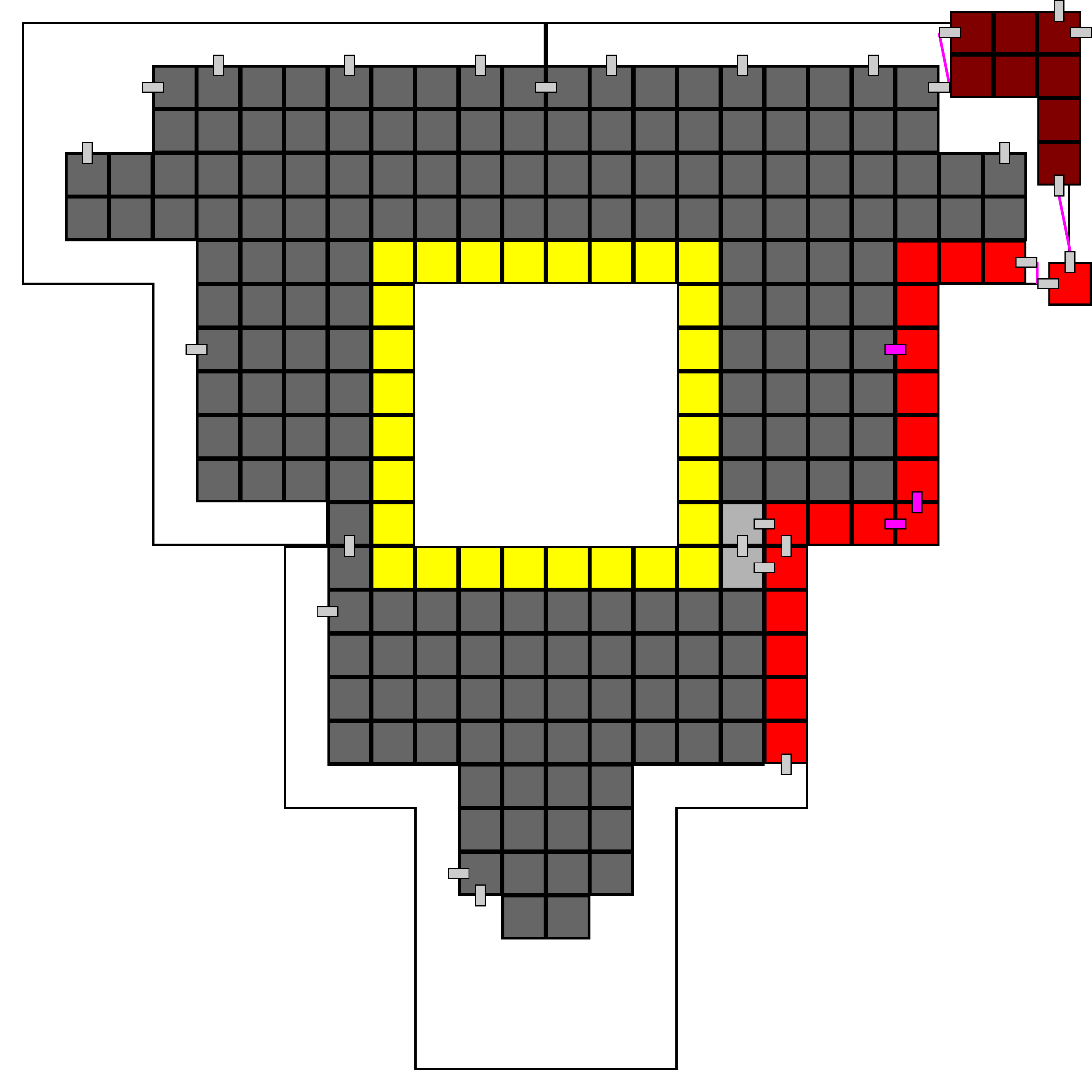}
		\label{4baseProcess_4}
	}
	\subfigure[Right fillers placed.] {
		\includegraphics[width=.30\textwidth]{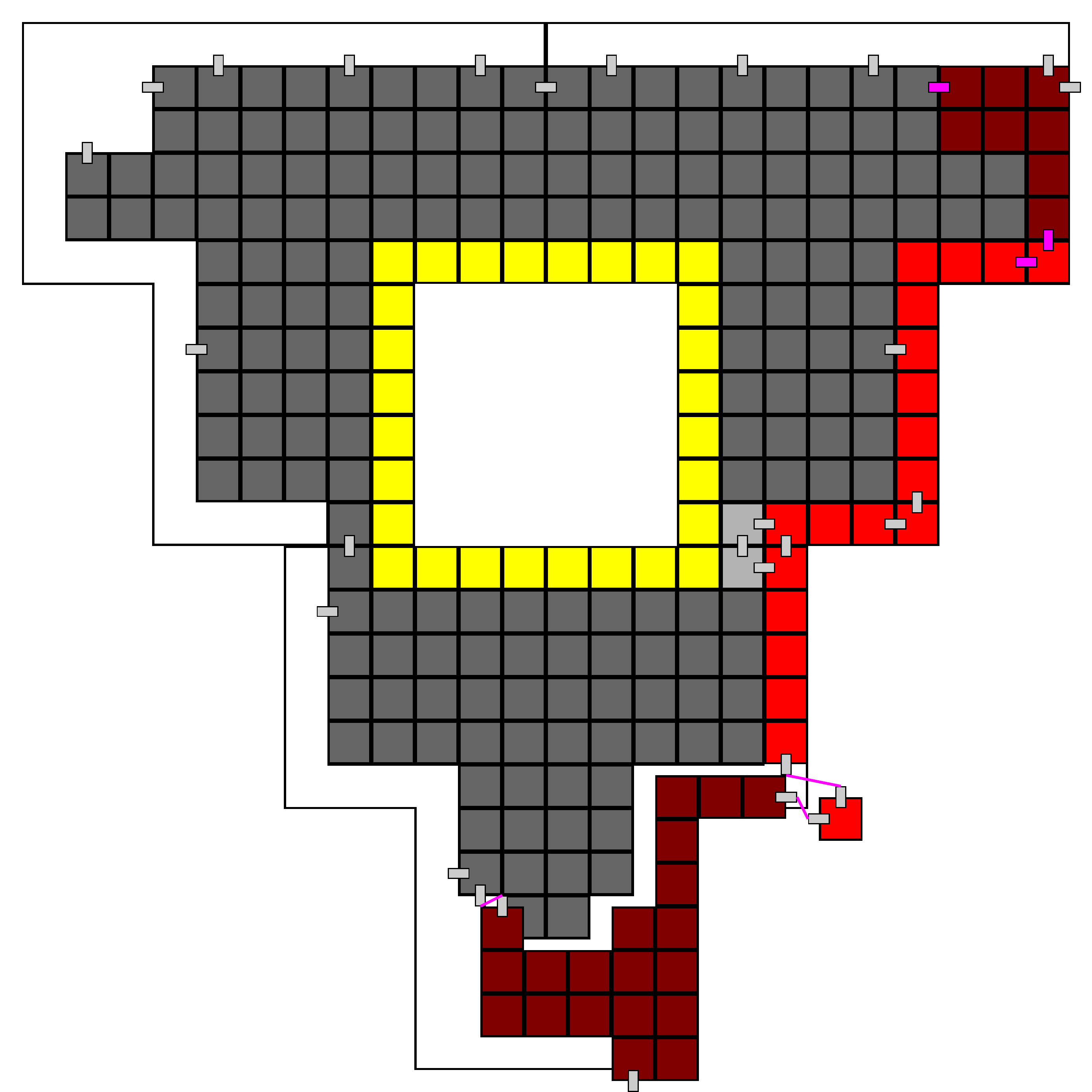}
		\label{4baseProcess_5}
	}
	\subfigure[A new combinable shape is produced.] {
		\includegraphics[width=.30\textwidth]{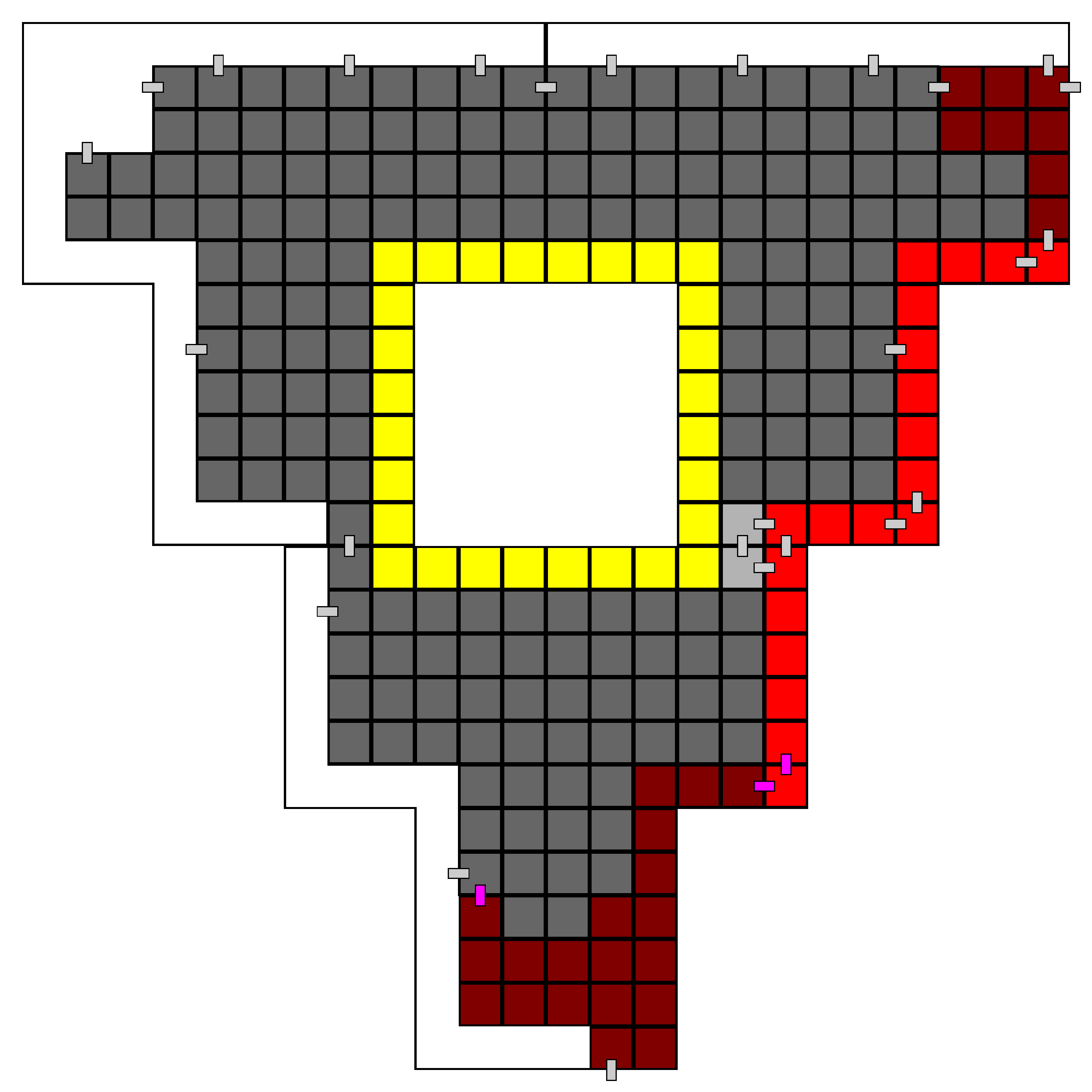}
		\label{4baseProcess_6}
	}
\caption{The process of keystone attachment and edge growth.}
\label{4baseProcess}
\end{center}
\end{figure}

\paragraph{Connector Glues.} When the keystone tiles have bonded to a base shape of $S^3_2$, they will place tiles along the edge of the assembly, matching the color of the keystone tiles that have attached. This process is shown in Figure \ref{4baseProcess_2}-\ref{4baseProcess_6}. \emph{Whitespace} is the area of the assembly that is part of the shape $S^3_i$ but is not covered by tiles. Each attachment of these tiles that fill the whitespace along the edge takes three hands to assemble, one for the base shape of $S^3_2$ (along with the already attached keystone and edge filling tiles) and two for the tiles that fill in the whitespace, which consists of a corner and the fitting assembly. The now combinable shape of $S^3_2$ will have exposed glues that will provide a method of attaching to the corresponding combinable shapes of $S^3_2$. This process is shown for a yellow base shape of $S^3_2$ in Figure \ref{4baseProcess}; in the case of the other base shapes of $S^3_2$, the blue base shape has a process that is a mirror of the shown yellow attachment process. The process for a red base shape consists of having the keystone attach atop the assembly to grow filling tiles along the top edge. The produced combinable shapes of $S^3_2$ have exposed strength-1 glues that will allow them to attach to two other combinable shapes of $S^3_2$, forming a base shape of $S^3_3$. Three combinable shapes of $S^3_2$ will only attach to one another if they have attached the same type of keystone tile, as shown in Figure \ref{4firstRed}.

\begin{figure}
\begin{center}
	\includegraphics[width=.75\textwidth]{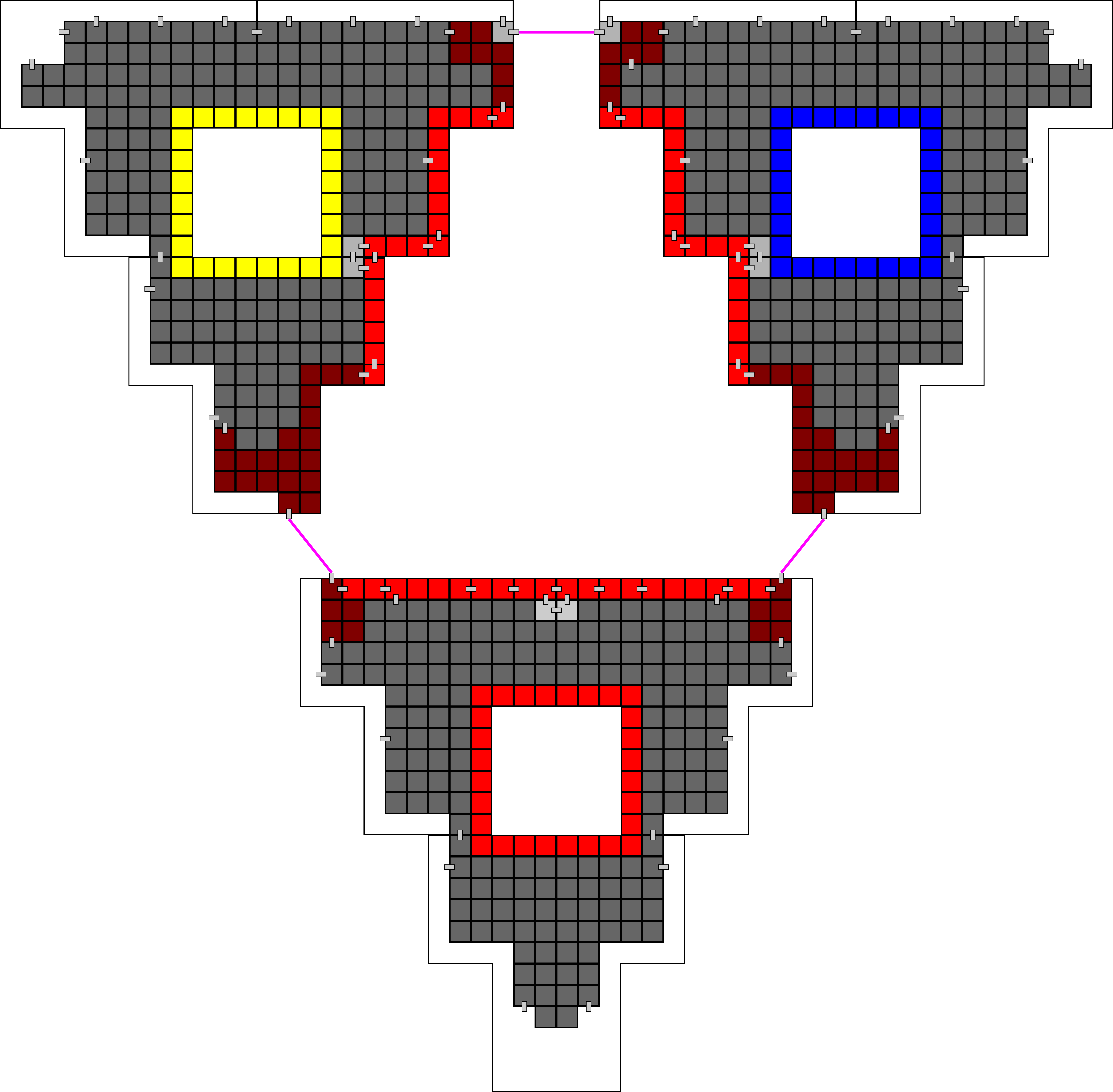}
	\caption{Different combinable shapes assemble into the next iteration of the Sierpinski triangle.}
	\label{4firstRed}
\end{center}
\end{figure}

\paragraph{Infinite Assembly.} A base shape of $S^3_i$ will form a combinable shape of $S^3_i$, and three combinable shapes of $S^3_i$ will form a base shape of $S^3_{i+1}$ in order to self-assemble the infinite Sierpinski triangle. If a base shape of $S^3_i$ was designated yellow (using keystone assemblies in stage $S^3_{i-1}$), it will place tiles along its right edge; if blue, it will place tiles along its left edge; if red, it will place tiles along its top edge. The filling tiles are designed to fill in arbitrarily long edges using a constant tileset. The three resulting red, blue, and yellow combinable shapes of $S^3_i$ will attach to form a base shape of $S^3_{i+1}$ if they have chosen the same keystone. The nondeterministic attachment of keystones allows for three combinable shapes of $S^3_i$ to produce three different base shapes of $S^3_{i+1}$, which will become complementary combinable shapes of $S^3_{i+1}$. These three combinable shapes of $S^3_{i+1}$ can produce base shapes of $S^3_{i+2}$. The nondeterminstic keystone attachment and arbitrary-length filling tiles allow us to self-assemble the discrete Sierpinski triangle using a 3-handed system at scale factor 3 with 990 tile types; 642 tiles for the base pieces and 348 tiles for the keystone assemblies and filler tiles.
\end{proof} 

%% file: C_2HAM.tex
\begin{theorem}\label{C_2HAM}
There exists a 2-HAM system $\Gamma = (T, \tau)$ that self-assembles $C_{\infty}^3$ with $|T| = 1,216$.
\end{theorem}

\begin{proof} \*
\paragraph{Overview.}
We show this by constructing a 2-HAM system $\Gamma = (T, 2)$. The tileset $T$ of the system contains 1,216 tiles, and self-assembles the discrete Sierpinski carpet at a scale factor of 3.  Intuitively, the construction works as follows. First, 8 base pieces are trivially assembled. Each base piece is nearly of shape $C_1^3$ except that each piece has a perimeter of width 1 missing from its outer edges so that tiles can be placed along the edges of the base pieces while staying consistent with the Sierpinski carpet shape. As seen in Figure \ref{summary_0}, the base pieces will place tiles along specific edges that expose glues allowing the base pieces to attach to one another. However, these tiles will be placed such that when the 8 base pieces come together, they form an assembly that is nearly of shape $C_2^3$ except that it is missing a perimeter of width 1 so that tiles can again be placed along its edges, allowing it to attach to other $C_2^3$ sized assemblies. These $C_2^3$ sized assemblies will attach to 7 other $C_2^3$ sized assemblies in a similar manner as the base pieces, to form an assembly that is nearly of shape $C_3^3$ except that, again, the shape will be missing a perimeter of width 1. This process will repeat indefinitely, allowing the system to self-assemble the infinite shape of the discrete Sierpinski carpet.

\begin{figure}
\begin{center}
	\includegraphics[width=.35\textwidth]{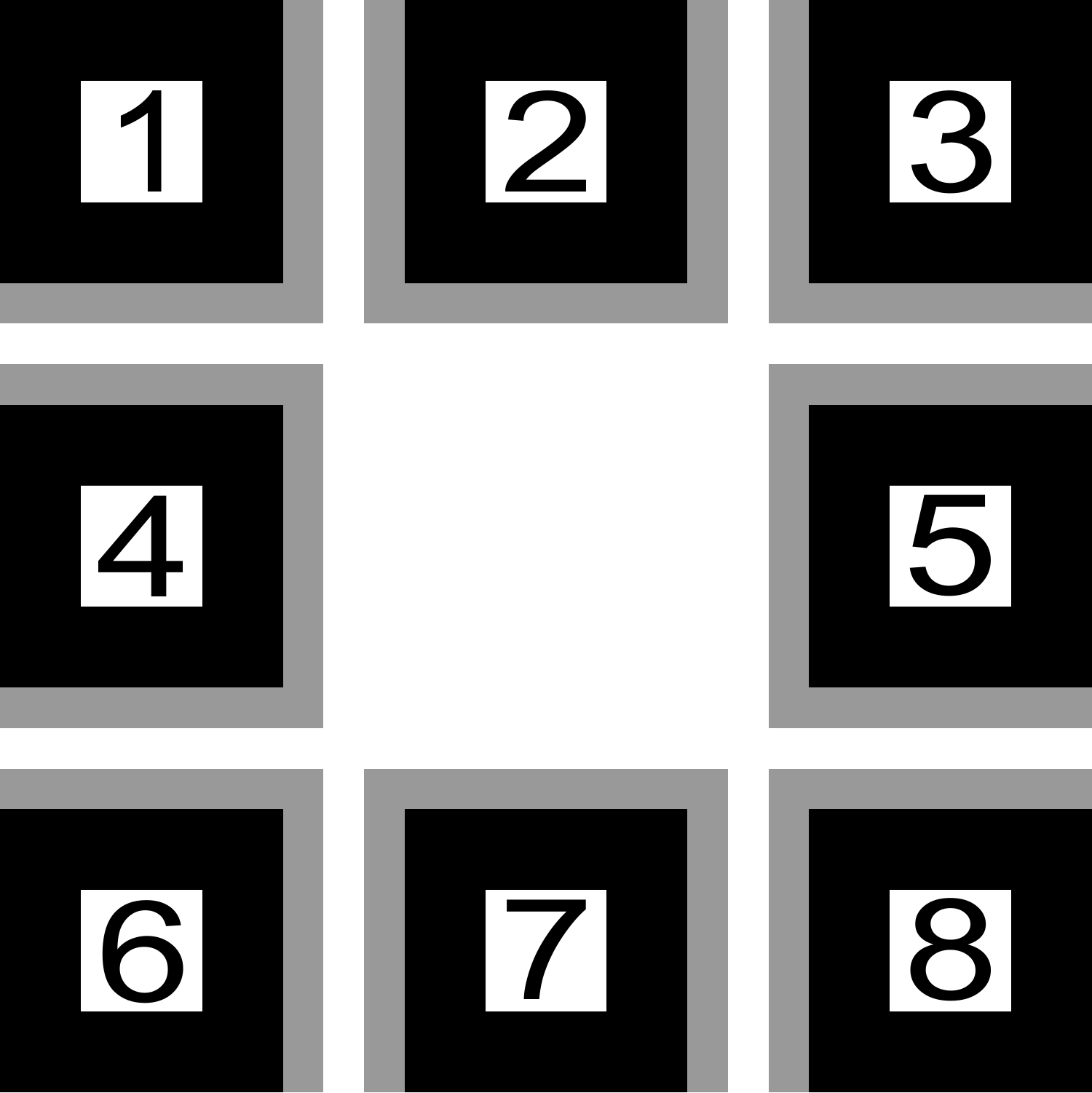}
	\caption{This figure shows the placement of tiles along the edges of the base pieces.}
	\label{summary_0}
\end{center}
\end{figure}

Our construction fundamentally relies on the nondeterminism inherent to tile assembly in order to achieve our result. Each base piece nondeterministically attaches to one of 8 different assemblies, called \emph{keystones}, that set off a chain of tile attachments that place glues in the appropriate positions. As shown in Figure \ref{summary_1}, the base piece labeled 1 can grow into any one of the 8 different assemblies shown. The other base pieces act similarly. The purpose of nondeterminism is not yet seen in this step; while the 8 different assemblies are geometrically similar at this stage, they pass along information that determines the behavior of the assembly that consists of 8 attached base pieces. In other words, they determine the behavior of the next iteration of the Sierpinski carpet.

\begin{figure}
\begin{center}
	\includegraphics[width=.75\textwidth]{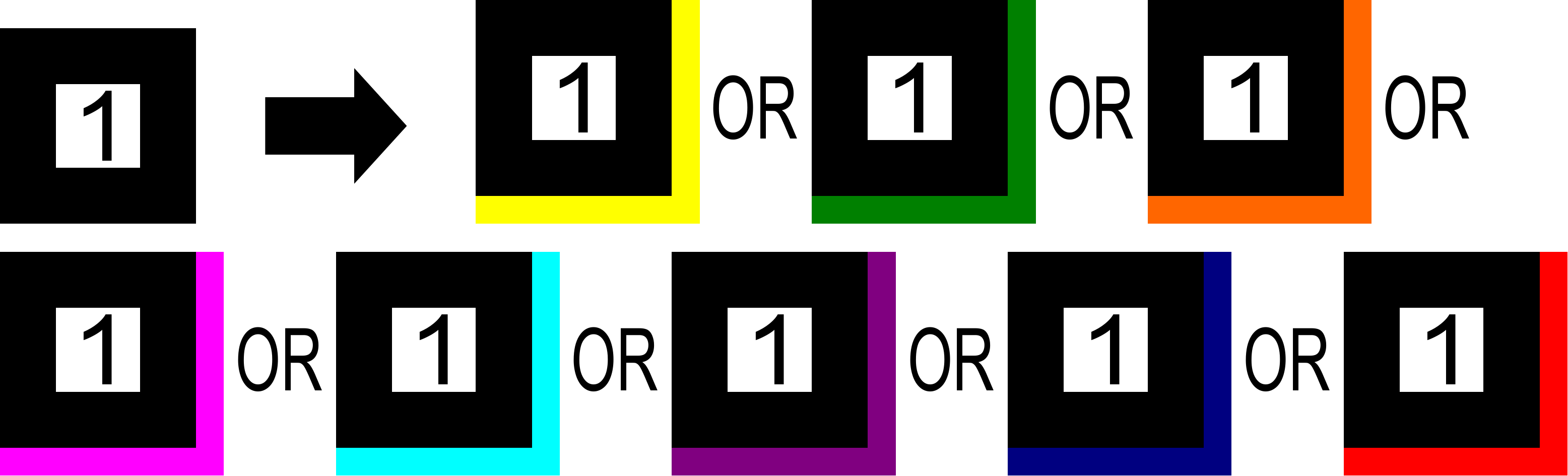}
	\caption{A base piece places tiles along specific edges, but will choose nondeterministically from 8 different tilesets to do so.}
	\label{summary_1}
\end{center}
\end{figure}

As shown in Figure \ref{summary_2}, 8 base pieces that have each nondeterministically attached the same \emph{keystones} will be able to attach to one another. If two base pieces attach different keystones, they will not attach to one another. Once 8 base pieces that have nondeterministically chosen the same keystones attach to each other, they will place glues along their edges according to the keystone they chose.

\begin{figure}
\begin{center}
	\includegraphics[width=.95\textwidth]{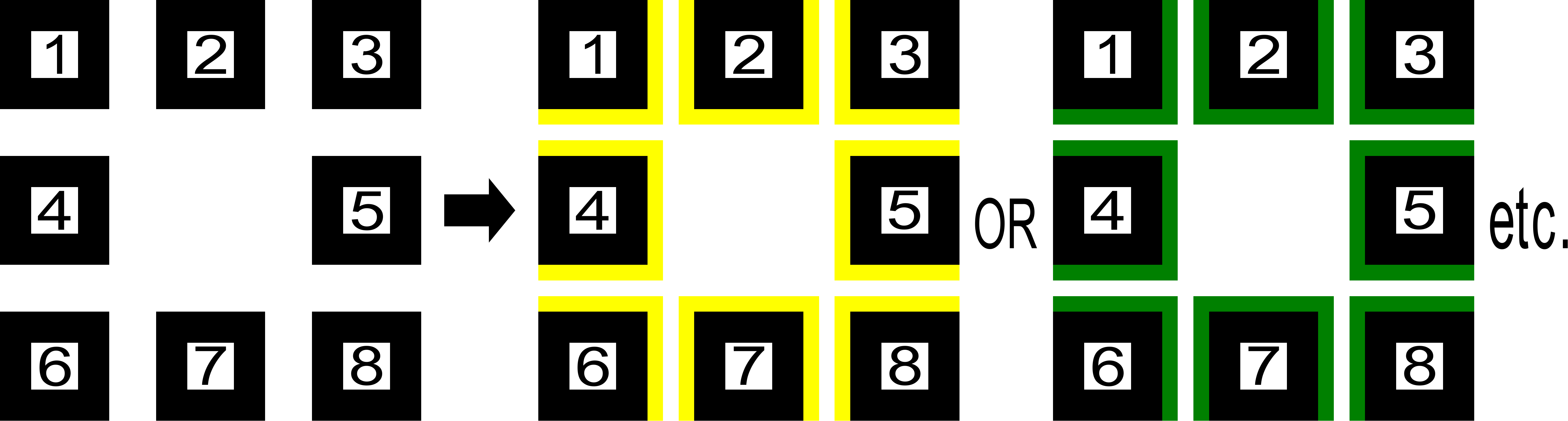}
	\caption{8 base pieces will attach to one another if they all nondeterministically attach the same keystone.}
	\label{summary_2}
\end{center}
\end{figure}

The keystone determines along which edges tiles are placed and which edges are left open, similar to Figure \ref{summary_0}. Figure \ref{summary_3} shows an assembly of 8 base pieces placing glues along their south and east edges as determined by their keystone. Note that the assembly can nondeterministically grow into one of 8 assemblies, again by the attachment of a keystone. The keystone assemblies attached at this step are the same as the ones in the base piece step, which allows our system a constant tile complexity. Figure \ref{summary_4} shows a different assembly of 8 base pieces, having attached different keystones than the base pieces in Figure \ref{summary_3}, thus placing tiles along its west, east, and south edges.

\begin{figure}
\begin{center}
	\includegraphics[width=.95\textwidth]{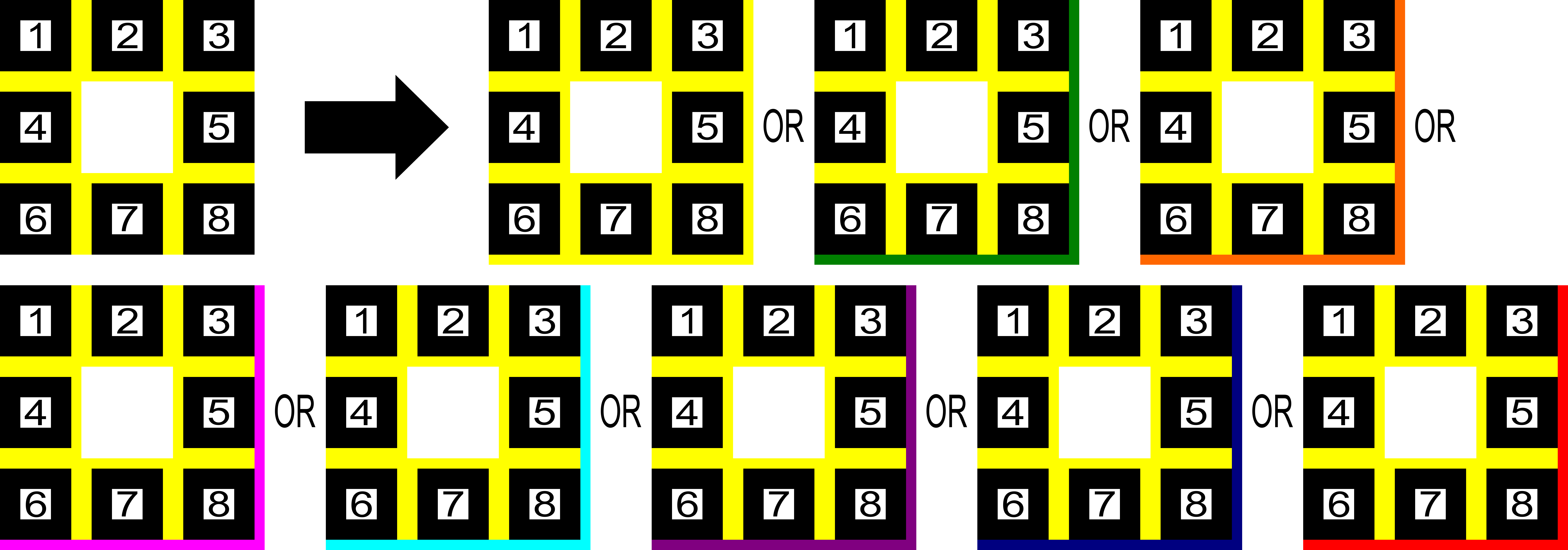}
	\caption{A set of 8 assemblies that have chosen the same keystone will attach and nondeterministically choose a keystone for the next iteration of the carpet.}
	\label{summary_3}
\end{center}
\end{figure}

\begin{figure}
\begin{center}
	\includegraphics[width=.95\textwidth]{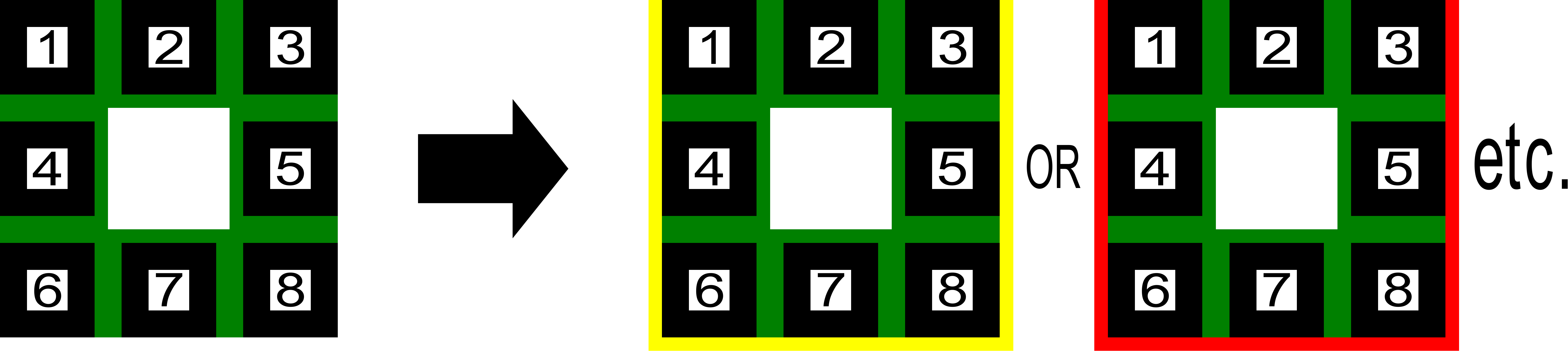}
	\caption{A group of 8 assemblies that have nondeterministically chosen to be the top-middle square in the Sierpinski carpet, placing their glues along the appropriate edges.}
	\label{summary_4}
\end{center}
\end{figure}

8 sets of 8 base pieces are shown in Figure \ref{summary_5}. In each set of 8 base pieces, each base piece has attached the same keystones at the first stage, but each set of 8 base pieces attached using different keystones than the other sets. These 8 sets will place tiles on the appropriate edges so that when the 8 sets attach, they will form the next iteration of the Sierpinski carpet with a missing perimeter of width 1. However, each set must choose the same keystone in order to attach to the other sets. 8 sets of 8 base pieces can grow nondeterministically into 8 different assemblies, again by the process of keystone attachment. The tile system is designed to allow the process of nondeterministic keystone attachment at any iteration of the discrete Sierpinski carpet, allowing for the finite assembly of the infinite shape.

\begin{figure}
\begin{center}
	\includegraphics[width=.95\textwidth]{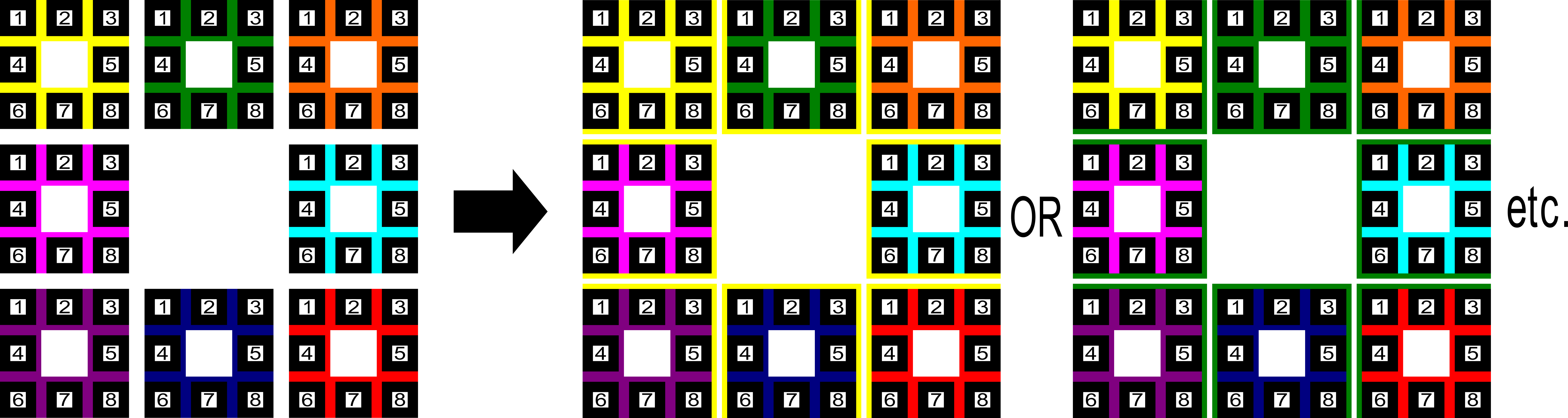}
	\caption{8 sets of 8 base pieces nondeterministically choose which role they will play in the next iteration of the process.}
	\label{summary_5}
\end{center}
\end{figure}

\paragraph{Base Shapes.} An assembly is considered to be a \emph{base shape of $C_i^3$} if it consists of only the points in $C_i^3$ minus an outer perimeter of width 1 with indented corners as shown in Figure \ref{C_baseShape}. The outer perimeter of a base shape of $C_i^3$ must expose two \emph{keystone glues} and otherwise expose \emph{base glues}; both of these glue types will be described in following sections. Note that we require a scale factor of 3 in order to assemble these base shapes. We will construct base shapes of size $C_1^3$, and then show that any base shape of $C_i^3$ can grow into a base shape of $C_{i+1}^3$. Also, there will be 8 distinct assemblies of each base shape of $C_i^3$, one for each \emph{position}. A \emph{position} is an integer 1 through 8 describing how one square-shaped assembly attaches to another, as shown in Figure \ref{summary_0}. An assembly at position 1 will attach assemblies on its east edge and south edge, an assembly at position 2 on its west and east edge, etc.

\begin{figure}
\begin{center}
	\includegraphics[width=.50\textwidth]{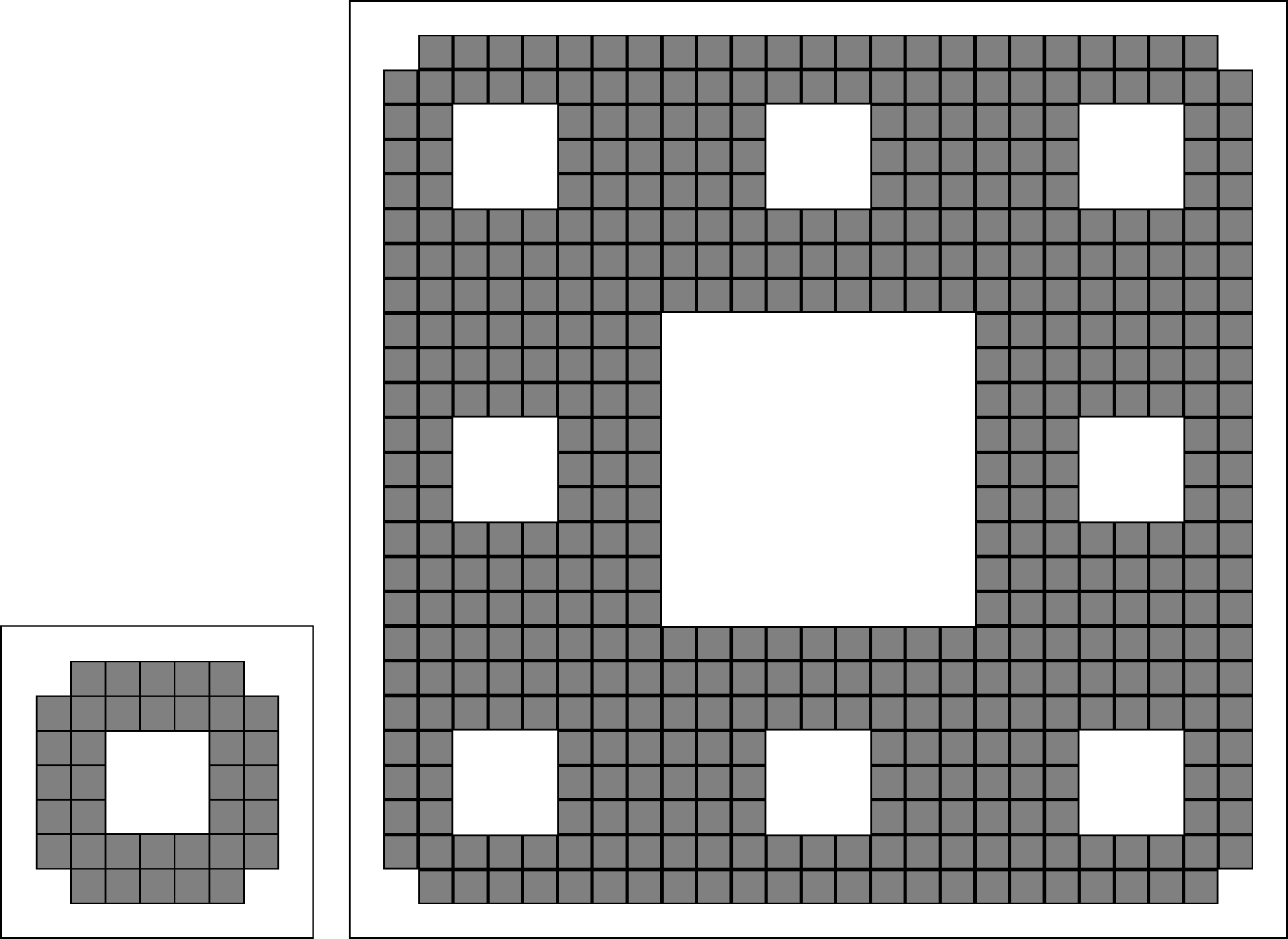}
	\caption{The points of a base shape of $C_1^3$ and base shape of $C_2^3$, respectively. An outline is shown to indicate the actual shape of $C_1^3$ and $C_2^3$.}
	\label{C_baseShape}
\end{center}
\end{figure}

\paragraph{Formation of Base Pieces.}
8 distinct base pieces of base shape $C_1^3$ are formed using 36 tiles each. Each base piece consists of two assemblies called \emph{base parts}. The base parts attach to one another using cooperative binding with glues at separate ends of the base parts, as shown in Figure \ref{base}. This ensures that each base part is completely assembled before it can attach to its corresponding base part to form a complete base piece. Along the edges of the base pieces are glues labeled with an N, W, S, or E followed by a B, corresponding to the cardinal directions. These glues will be referred to as \emph{base glues}, and their purpose will be explained in the \emph{Connector Glues} section.

\begin{figure}
\begin{center}
	\includegraphics[width=.90\textwidth]{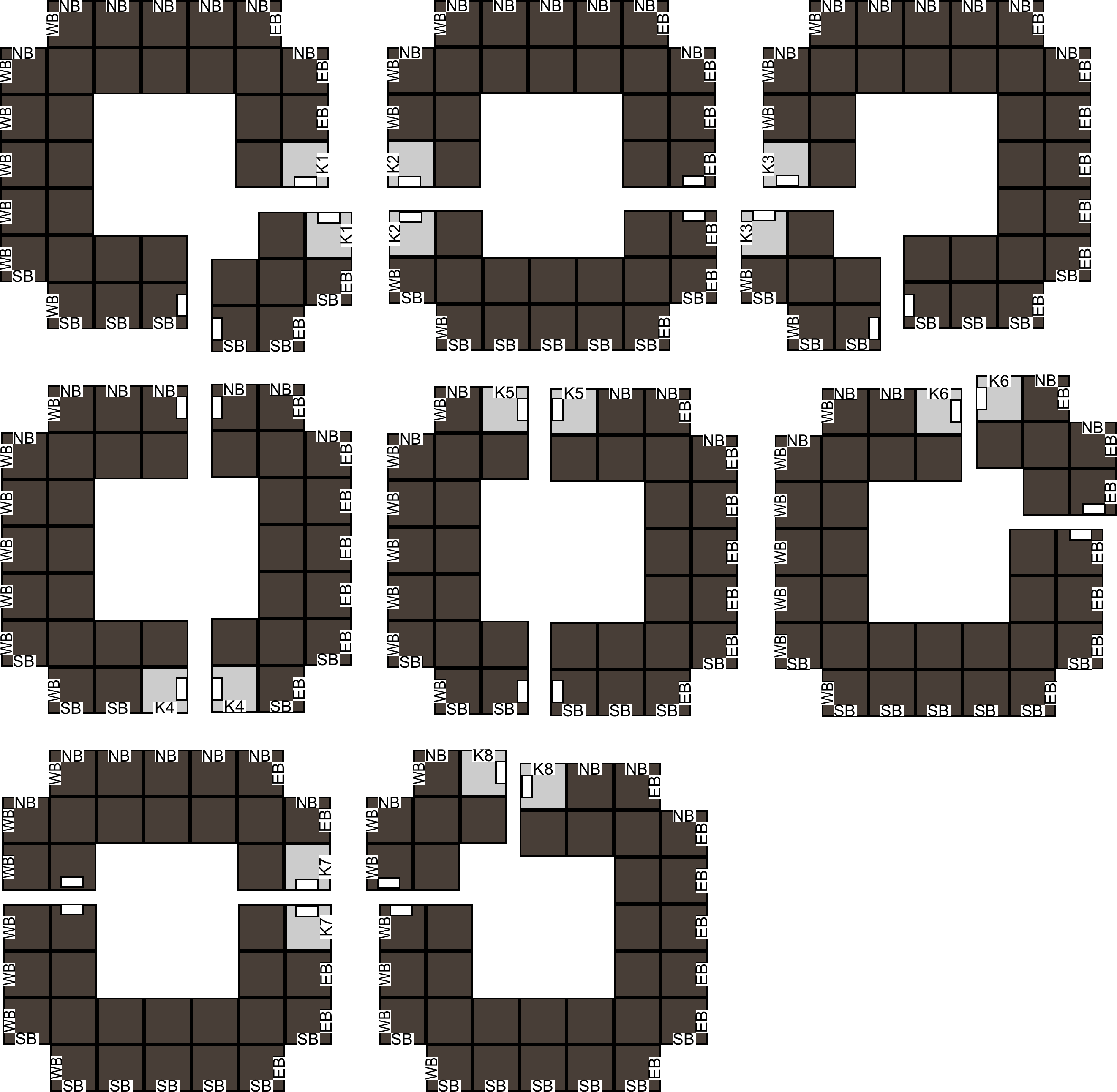}
	\caption{The distinct base pieces and their base part attachments.}
	\label{base}
\end{center}
\end{figure}

\paragraph{Keystone Attachments.}
Each of the 8 base pieces has its own set of 8 keystones, for a total of 64, shown in Figure \ref{keyFigures}. A \emph{keystone} is a two-tile assembly designed to attach to a base piece once the base piece is complete. Any one of the 8 keystones corresponding to a particular base piece is able to attach to the base assembly in a nondeterministic fashion.

\begin{figure}
\begin{center}
	\subfigure[] {
		\includegraphics[width=.45\textwidth]{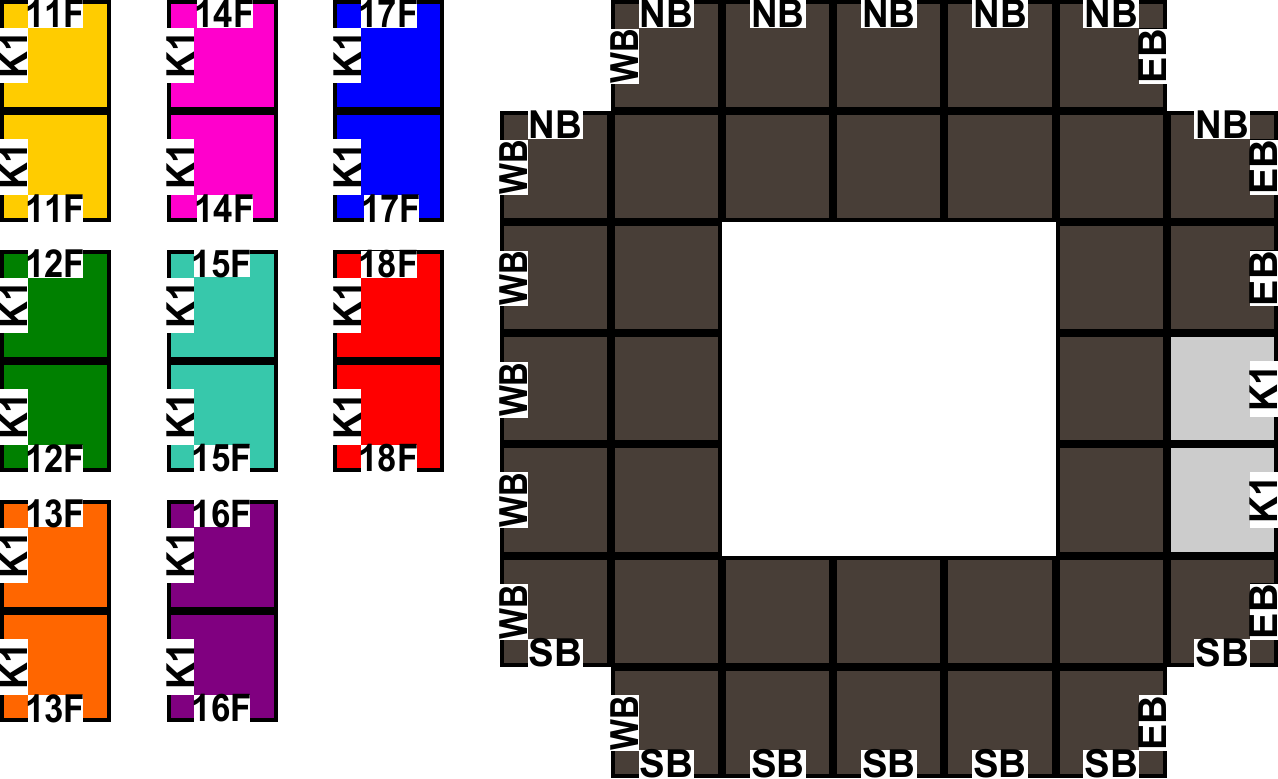}
		\label{keyFigures_1}
	}
	\subfigure[] {
		\includegraphics[width=.45\textwidth]{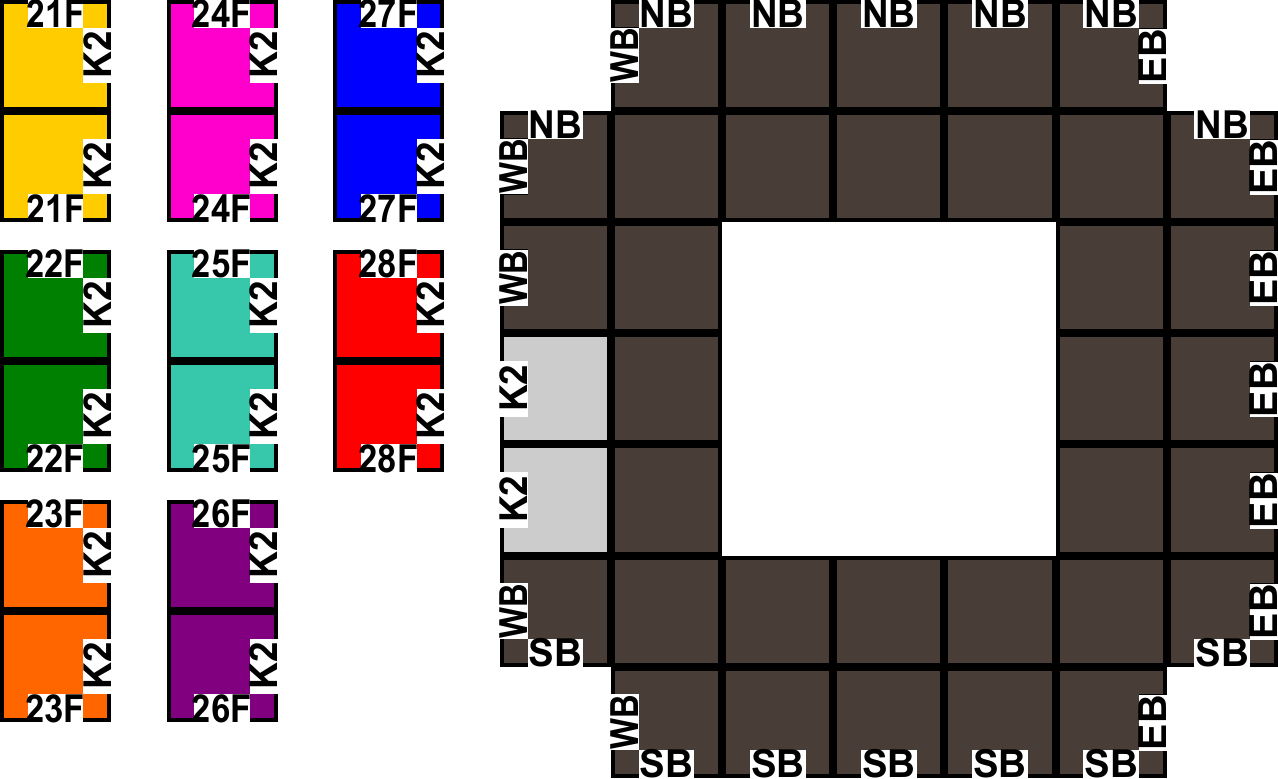}
		\label{keyFigures_2}
	}
	\subfigure[] {
		\includegraphics[width=.45\textwidth]{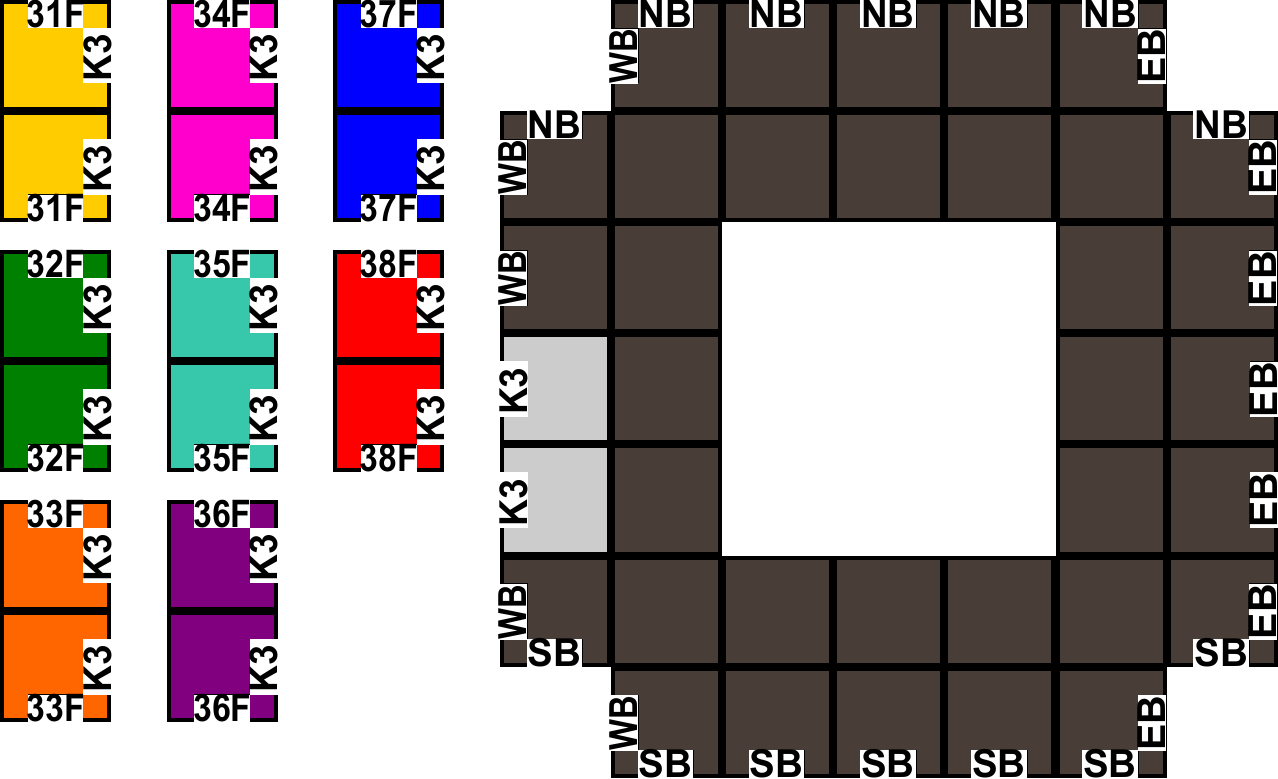}
		\label{keyFigures_3}
	}
	\subfigure[] {
		\includegraphics[width=.45\textwidth]{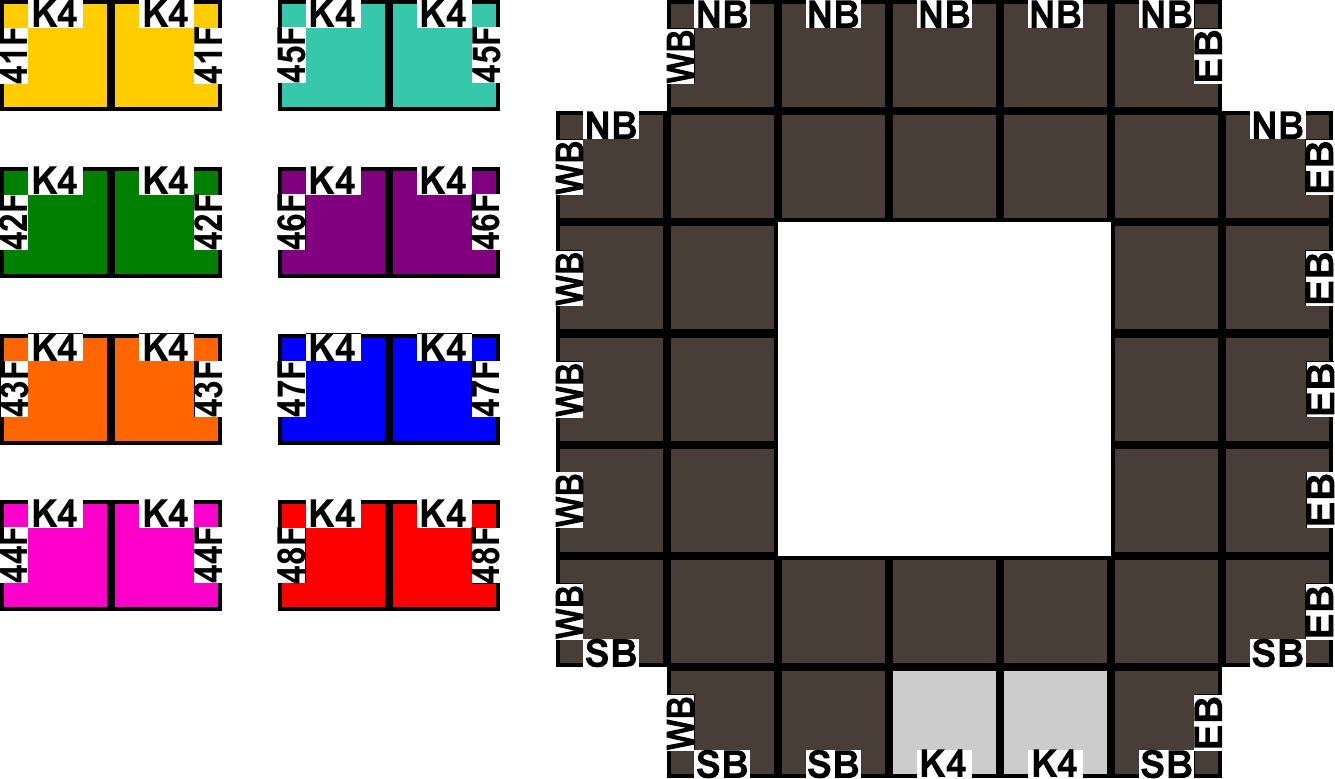}
		\label{keyFigures_4}
	}
	\subfigure[] {
		\includegraphics[width=.45\textwidth]{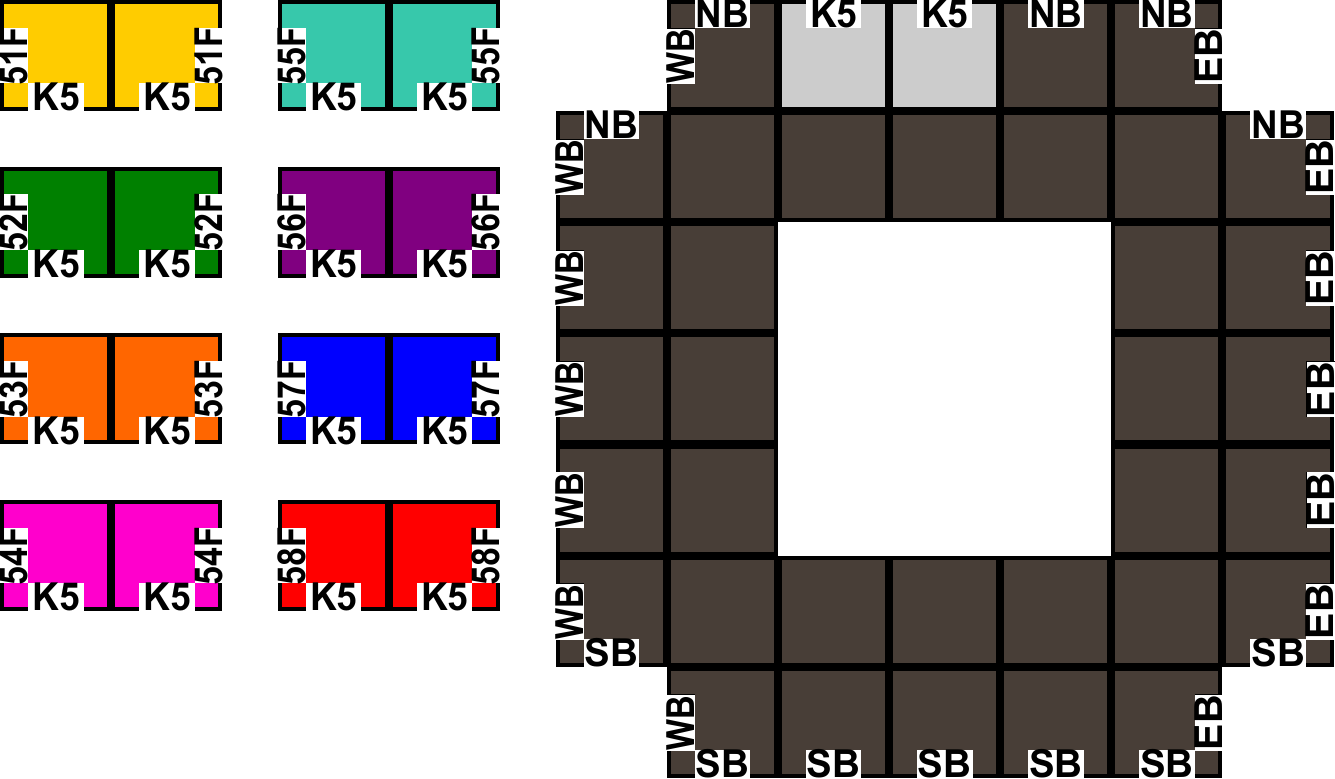}
		\label{keyFigures_5}
	}
	\subfigure[] {
		\includegraphics[width=.45\textwidth]{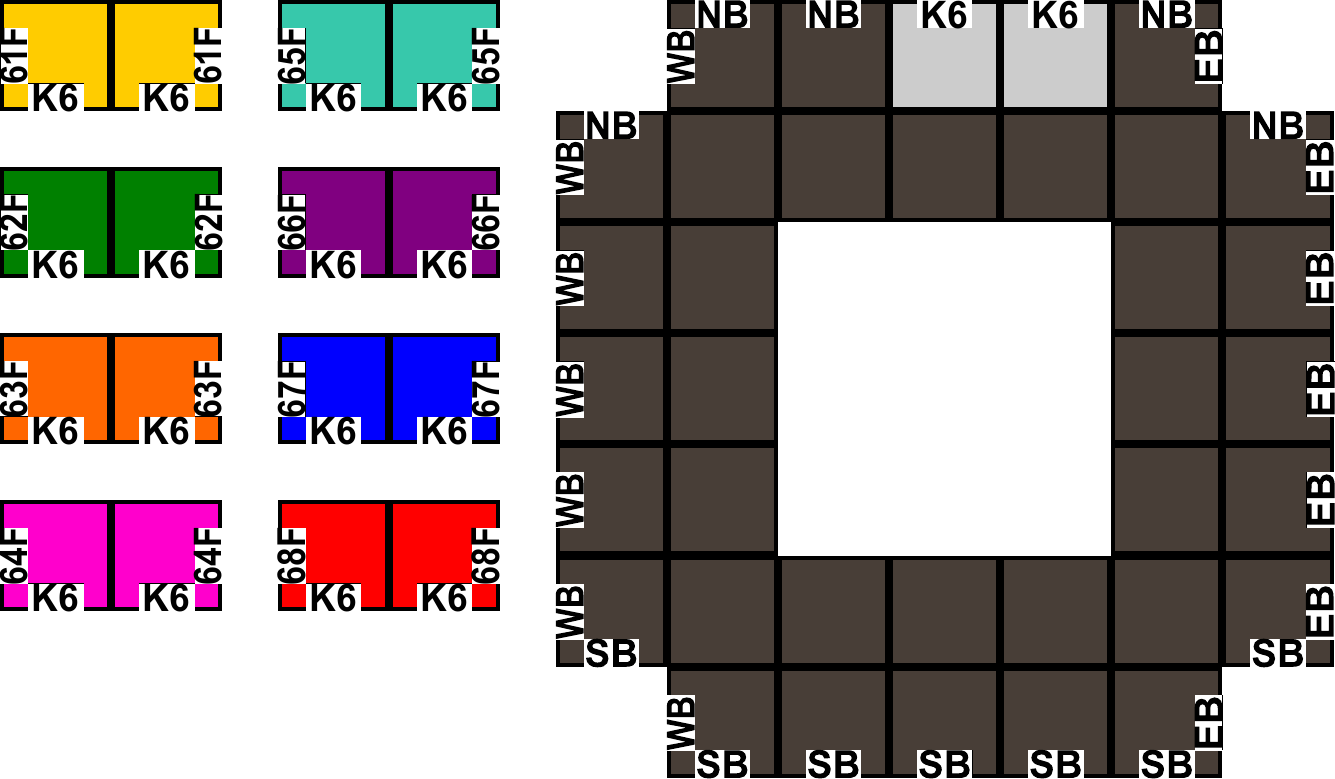}
		\label{keyFigures_6}
	}
	\subfigure[] {
		\includegraphics[width=.45\textwidth]{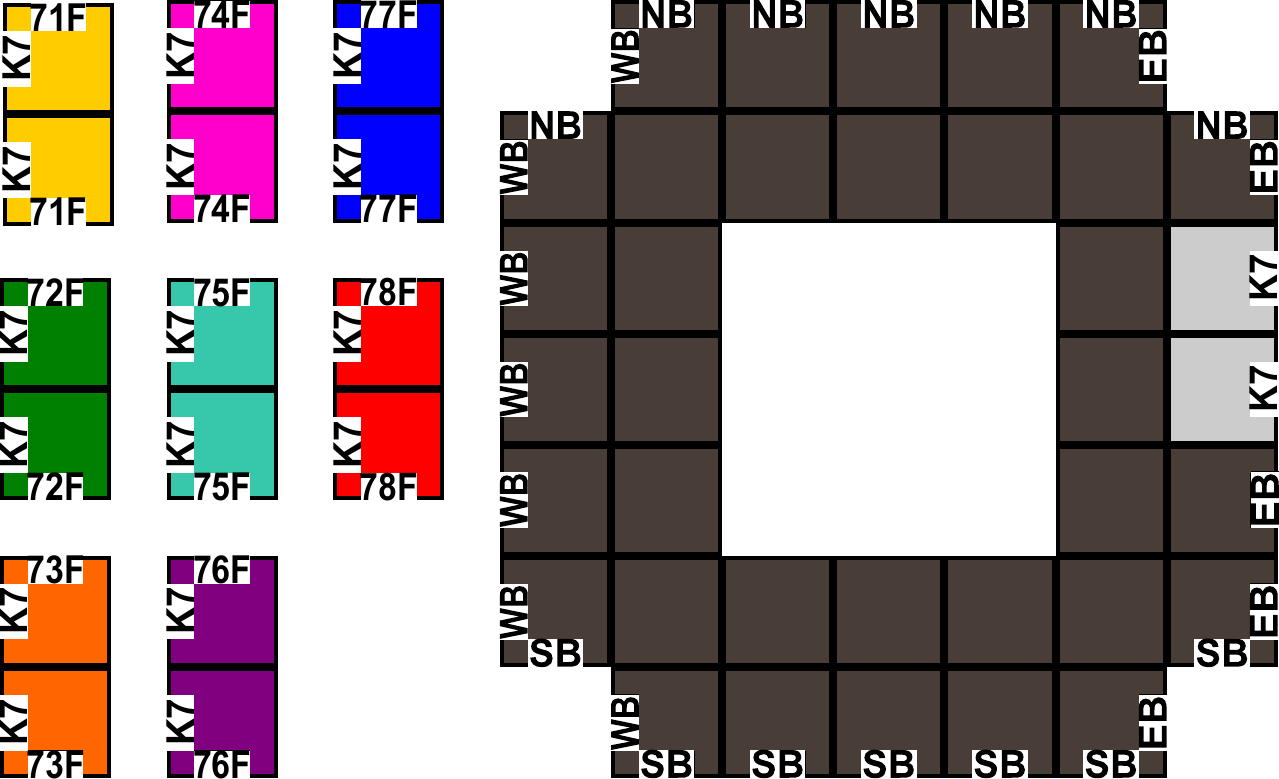}
		\label{keyFigures_7}
	}
	\subfigure[] {
		\includegraphics[width=.45\textwidth]{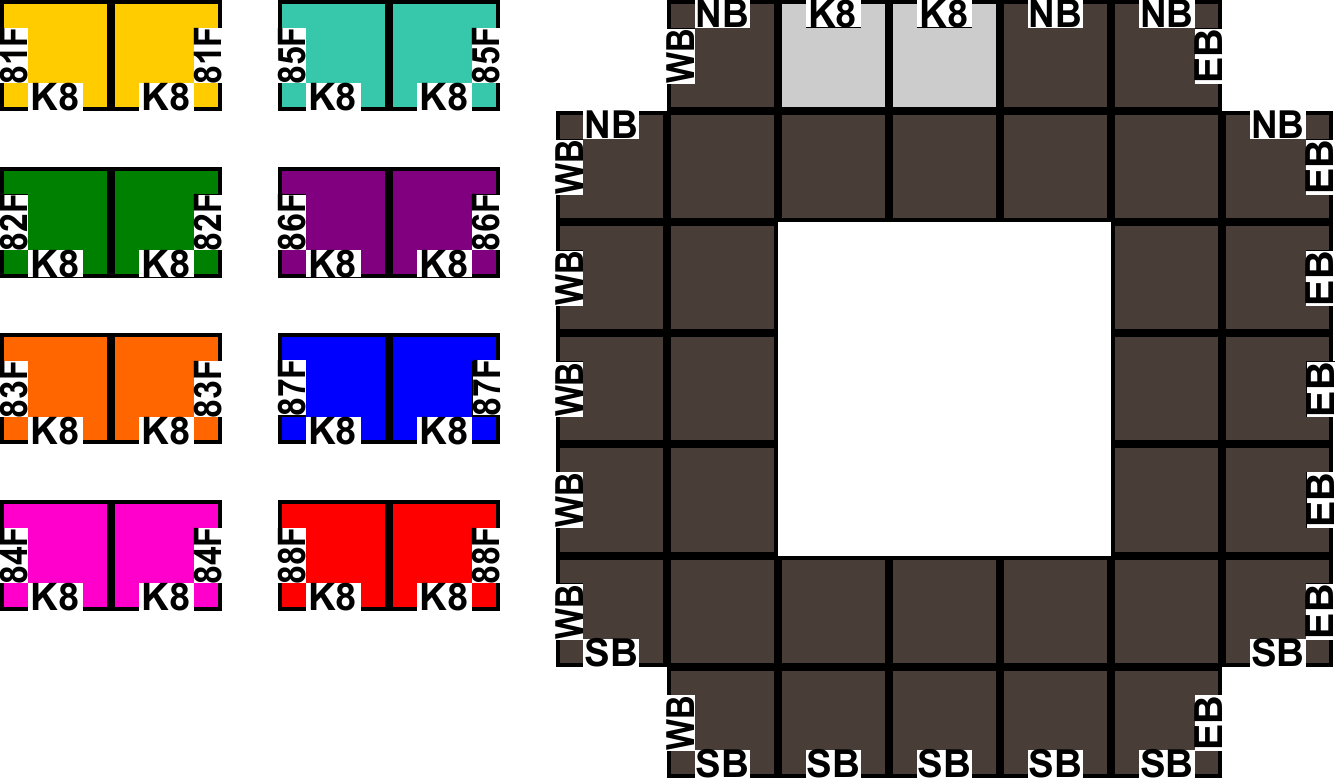}
		\label{keyFigures_8}
	}
\caption{The base pieces and their corresponding keystone assemblies.}
\label{keyFigures}
\end{center}
\end{figure}

Note that a base piece must complete before a keystone can attach, since each keystone binds cooperatively to the base piece. On the keystone assemblies and the base pieces, there are glues labeled K followed by a number corresponding to that base piece's position, ensuring that the keystones attach to the appropriate base pieces. We will refer to these glues as the \emph{keystone glues}. The keystone glues exposed on the base piece are not available for cooperative binding until the base piece completes. Similarly, the pair of tiles that make up the keystone assemblies must also bind to one another before they can attach to the base piece.

\paragraph{Connector Glues.}
The keystones initiate a growth of tiles, called \emph{filler tiles}, along the edges of the base pieces. The filler tiles are designed with glues labeled with two numbers and a letter F, which will hereby be referred to as the \emph{filler glue}. The first number of a filler glue refers to the position of the base piece that the filler tile attaches to. The second number refers to a \emph{keystone type}. A keystone type determines the position of the assembly consisting of 8 base pieces that have attached corresponding keystone types; in other words, when a keystone attaches to a base shape of $C_n^3$, the keystone type determines the position of the base shape of $C_{n+1}^3$ in the base shape of $C_{n+2}^3$.

A filler tile is placed cooperatively using the base glue exposed on the base piece and the filler glue exposed on the keystone. Once a filler tile attaches to the keystone and the base piece, another filler tile can be cooperatively attached using the filler glue that is exposed by a filler tile that has already attached to the assembly and the base glue exposed on the base piece. This process, shown in Figure \ref{filler}, can occur along arbitrarily long edges that expose the appropriate base glues. Note that the filler tiles can attach until they reach an indented corner of a base piece, since there is no base glue available for cooperative binding.

\begin{figure}
\begin{center}
	\includegraphics[width=.95\textwidth]{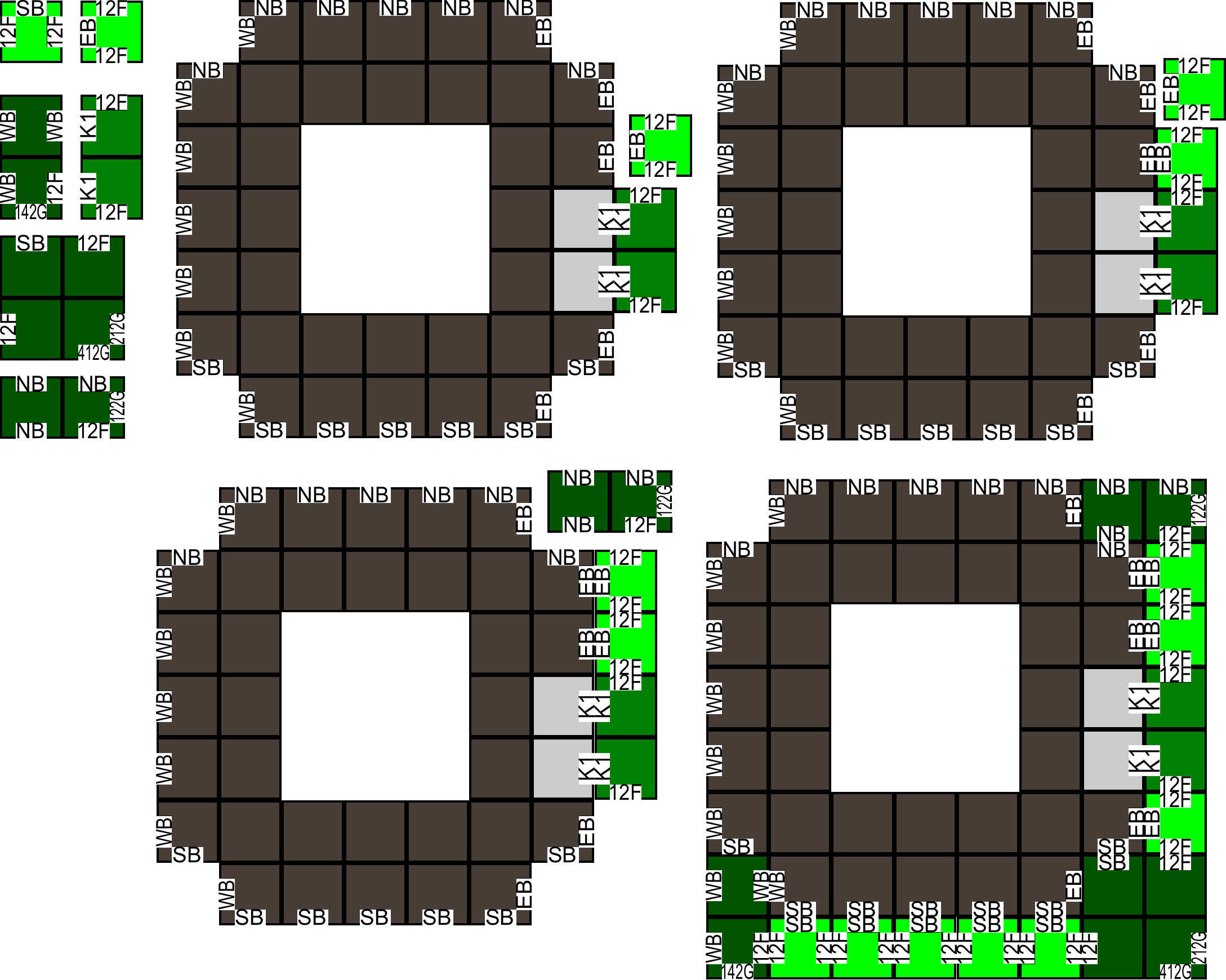}
	\caption{The process of filler tile attachment, leading to the attachment of corner assemblies.}
	\label{filler}
\end{center}
\end{figure}

The \emph{corner assemblies} are a set of assemblies geometrically designed to cooperatively attach using the filler glue exposed by a filler tile and the base glue exposed on the inside of an indented corner of a base piece, as shown in Figure \ref{filler}. These \emph{corner tiles}, once attached to a base piece, expose \emph{connector glues} that allow for the base piece to attach to another base piece. A connector glue is a glue labeled with three numbers followed by a G. The first two numbers refer to the positions of the base pieces that are being attached to one another by the binding of the connector glues. The third number refers to the keystone type. Note that the base pieces will only attach to one another if they have attached the same keystone types, as the keystone type information is passed through the filler tiles to the corner tiles and their connector glues.

\begin{figure}
\begin{center}
	\includegraphics[width=.95\textwidth]{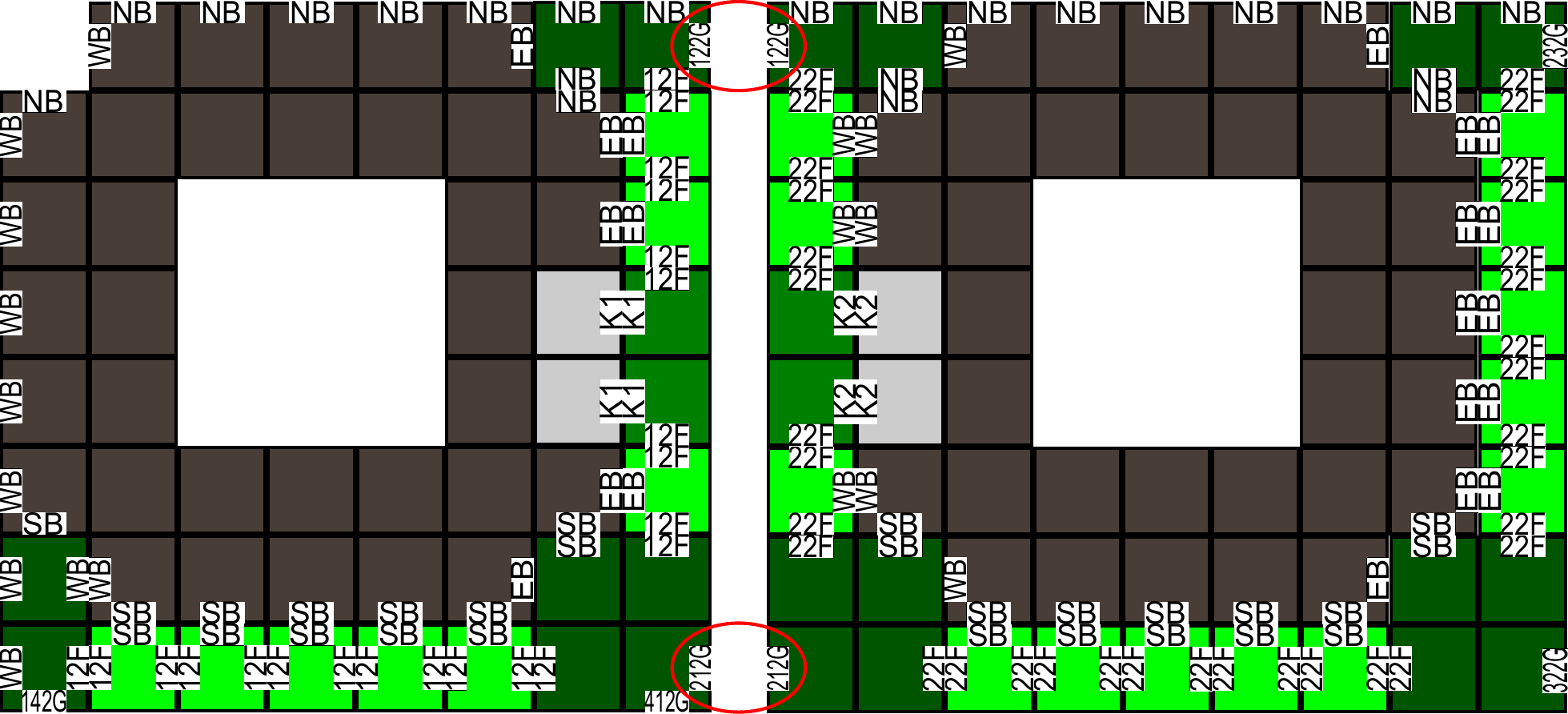}
	\caption{Two adjacent positioned assemblies attach using their connector glues.}
	\label{connectors}
\end{center}
\end{figure}

Assemblies that have attached sufficient connector glues attach to other assemblies that have been attached with the same \emph{keystone type} and have also attached sufficient connector glues. The connector glues attach two assemblies using cooperative binding at opposite ends of the adjacent sides of the assemblies. See Figure \ref{connectors}.

\begin{figure}
\begin{center}
	\includegraphics[width=.95\textwidth]{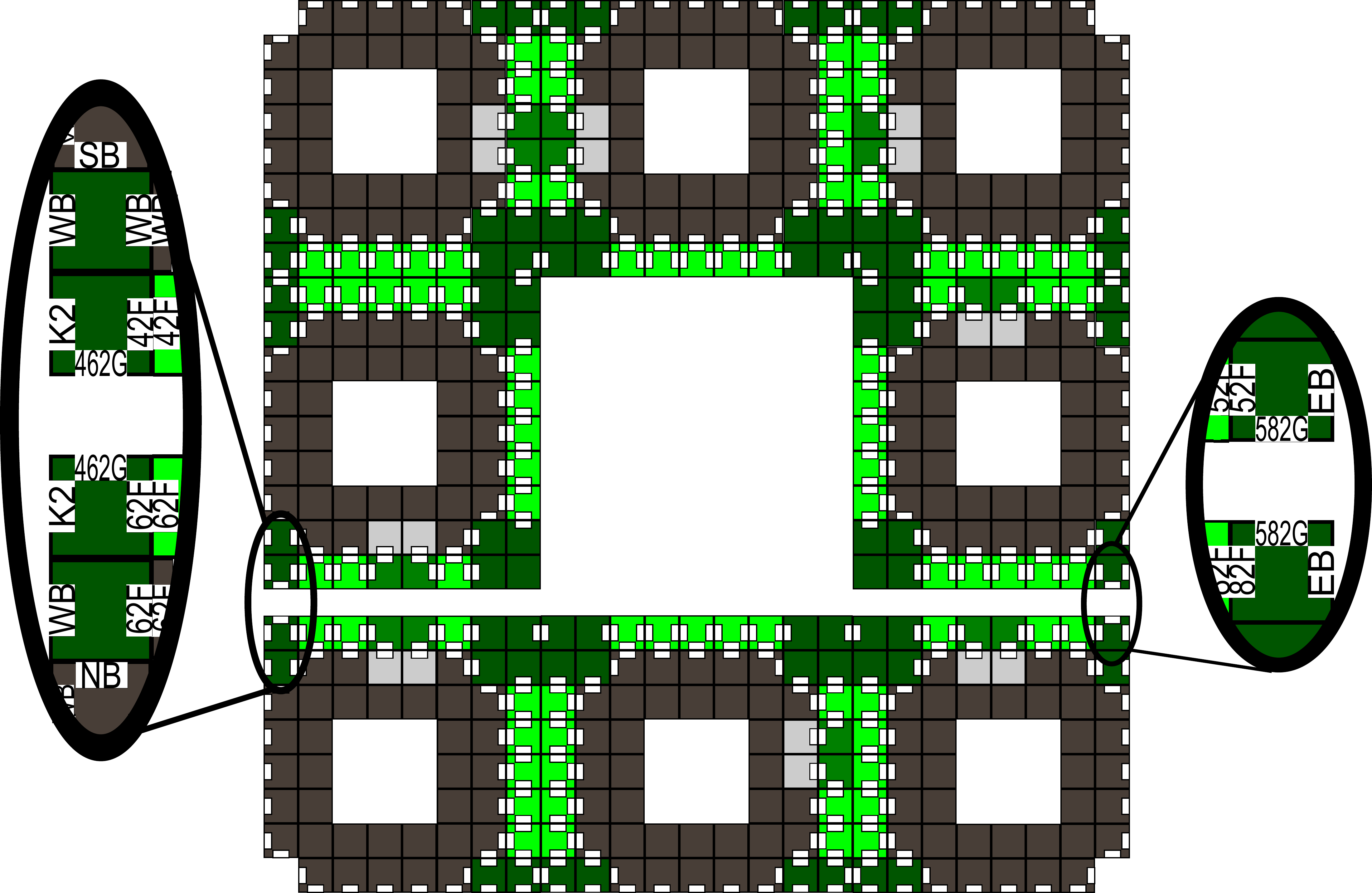}
	\caption{The creation of a base shape of $C_2^3$, showing its newly exposed keystone glues.}
	\label{coop}
\end{center}
\end{figure}

For each keystone type, there is an assembly at a position or group of positions that is designed to attach cooperatively using two adjacent positioned assemblies, or in the case of positions 2, 4, 5, and 7, using three positioned assemblies. For example, an assembly with keystone type corresponding to position 2 will attach its subassembly corresponding to positions 6, 7, and 8 cooperatively, only after placing subassemblies at positions 4 and 5, as shown in Figure \ref{coop}. Each of these assemblies imitates its base piece of corresponding position in terms of cooperative binding of the final subassembly. Note the similarities in the final attachment in the 2nd base piece in Figure \ref{base} and Figure \ref{coop}. Once these final subassemblies are attached, the resulting assemblies are base shapes of $C_2^3$. At the point of these final cooperative attachments, two keystone glues are exposed, allowing for the attachment of another keystone, repeating the keystone attachment process, followed again by the filler tiles and then their corner assemblies. In the example shown in Figure \ref{coop}, the assembly will be in position 2 in the creation of a base shape of $C_3^3$. A finite set of tiles is used to repeat this process indefinitely. We use 288 tiles for the base pieces, 128 tiles for the keystone assemblies, and 800 tiles for the filler tiles and corner assemblies, for a total of 1,216 tiles.
\end{proof} 

%% file: impossibility.tex
\section{Impossibility Results}
\label{sec:impossibility}

\input{aTAM_impossibility}

\input{near_perfect_impossibility}

%% file: aTAM_impossibility.tex
\subsection{The Sierpinski Triangle Does Not Strictly Assemble in the aTAM}
In this section we show that $S^c_{\infty}$ is not strictly self-assembled by any aTAM system.  Our argument is based on the constant width choke points that separate different iterations of the Sierpinski triangle.  Our proof simply pigeon-holes over these choke points to generate a reoccurring tiling (with respect to tile types placed, order in which points are tiled, and even neighbors present when a tile is attached) to generate a contradiction.  To prove our impossibility result, we first provide a formal definition of a reoccurring tiling called a \emph{window-movie} (we borrow/abuse the terminology from~\cite{MPS2014UTS}).

Our definitions of windows and window-movies are taken from \cite{MPS2014UTS}.  However, our \emph{windows} are slightly more general in that they may be any set of edges of the 2D lattice grid, even a finite set, as opposed to only allowing for sets that partition the 2D lattice into two disjoint regions.

\begin{definition}[Window] A window is any subset of edges taken from the grid graph, which is defined by vertex set $\mathbb{Z}^2$ and edge set consisting of point pairs of exactly unit distance.

The next concept used in our result is the \emph{window-movie}.  A window-movie, with respect to a given window and an assembly sequence, is simply a specification of which glues appear along the edges of the window, and in what order, according to the given assembly sequence.  Further, two window-movies are said to be \emph{equivalent} if their respective windows are translations of each other, and if the glues placed on the window are the same in type and in order of placement.  The formal definition as presented by Meunier, Patitz, Summers, Theyssier, Winslow, and Woods~\cite{MPS2014UTS} is included below.

\begin{definition}[Window-Movie~\cite{MPS2014UTS}]
Given an assembly sequence $\overrightarrow{\alpha}$ and a window $w$, the associated window movie is the maximal sequence
\begin{equation}
M_{\overrightarrow{\alpha},w} = (v_0,g_0), (v_1,g_1), (v_2,g_2), \ldots
\end{equation}
of pairs of grid graph vertices $v_i$ and glues $g_i$, given by the order of the appearance of the glues along window $w$ in the assembly sequence $\overrightarrow{\alpha}$.  Furthermore, if $k$ glues appear along $w$ at the same instant (this happens upon placement of a tile which has multiple sides touching $w$), then these $k$ glues appear contiguously and are listed in lexicographical order of the unit vectors describing their orientation in $M_{\overrightarrow{\alpha},w}$.  Two window-movies are said to be equivalent if they are equal up to translation.
\end{definition}

\end{definition}


Our windows and window-movies will be focused on the following ``choke point" sets of edges within the Sierpinski triangle which specify constant width-$c$ regions that connect a given size iteration of the Sierpinski triangle to two equal size iterations through north/south bonding edges.

\begin{definition}
For a given scale factor $c$, define $\texttt{choke}^\ell_i$ (left choke) to be the set of $c$ edges of the 2D square lattice that connect the points of $S^c_i$ with $S^c_i +(-c2^i,c2^{i+1})$, and $\texttt{choke}^r_i$ (right choke) to be the set of $c$ edges of the 2D lattice that connect the points of $S^c_i$ with $S^c_i +(c2^i,c2^{i+1})$.  See Figure~\ref{aTAM_scaleChokes} for an example.
\end{definition}

\begin{figure}
\begin{center}
	\includegraphics[width=.75\textwidth]{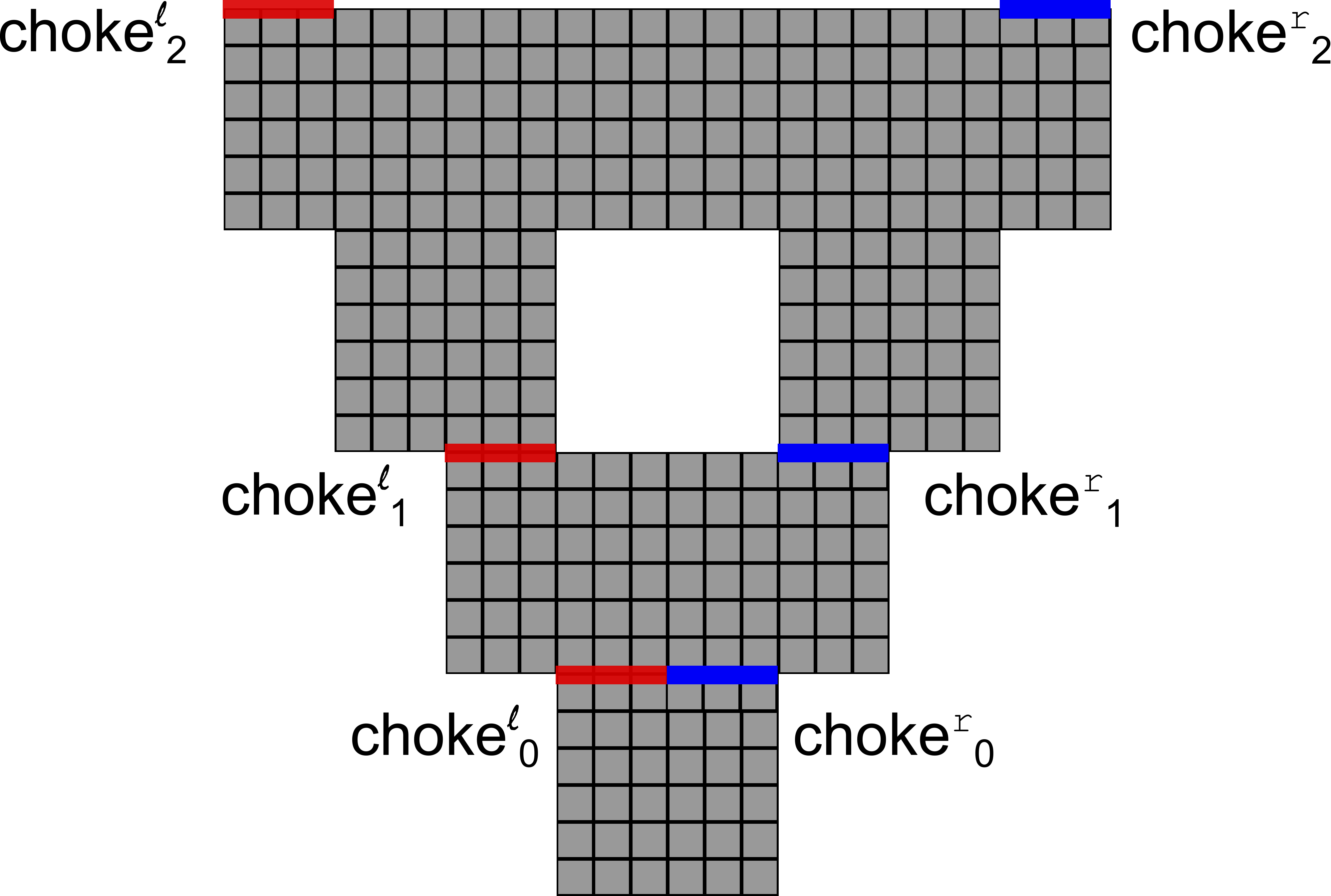}
	\caption{The left and right choke point edge sets for $S^3_{\infty}$ are shown. Edges in $\texttt{choke}^\ell_i$ are shown in red, while edges in $\texttt{choke}^r_i$ are shown in blue.}
	\label{aTAM_scaleChokes}
\end{center}
\end{figure}

Finally, our proof uses the following technical lemma which states that certain translations of pairs of points are such that they cannot both be contained within the Sierpinski triangle shape.

\begin{lemma}\label{lemma:pointlanding}
Let $c$, $i$, $j$, and $k$ be positive integers with $k,j>i\geq 1$, and let $\vec{v} = (c2^j-c2^i,c2^{j+1} - c2^{i+1})$. For any $p_y \in \{c2^{k+1} , c2^{k+1} +1, \ldots, c2^{k+1}+2c-1\}$, let $p^\ell = (-1,p_y)$ and $p^r = (0,p_y)$.  Then the translations of $p^\ell$ and $p^r$ by vector $\vec{v}$, $p^\ell_{\vec{v}}$ and $p^r_{\vec{v}}$, are such that either: 1) $p^\ell_{\vec{v}} \notin S^c_{\infty}$ or 2) $p^r_{\vec{v}} \notin S^c_{\infty}$.
\end{lemma}
\begin{proof}
First, consider the set of points defined by translating $p^\ell$ and $p^r$ by any positive integer multiple of the vector $(-c,2c)$, which includes the translation by vector $\vec{v}$.  The only such translations in which both points lie within $S^c_{\infty}$ are those that move the point to a $y$ position of height between $dc2^{k+1}$ and $dc2^{k+1} + 2c  -1$  for some positive integer $d$, i.e., any integer multiple of the $2^{k}$ high block of size $2c$.  However, any translation of $p^\ell$ and $p^{r}$ by $\vec{v} = (c2^j - c2^i,c2^{j+1} - c2^{i+1})$  will place the points at $y$-ordinate $p_y + c2^{j+1} - c2^{i+1}$.  As $k>i$ and $j>i$, we know that $p_y + c2^{j+1} - c2^{i+1}$ cannot fall within the required range, i.e., it cannot sum to an integer multiple of $c2^{k+1}$ plus some $\Delta$ between $0$ and $2c -1$.
\end{proof}

\begin{theorem}\label{thm:aTAMImpossible}
$S_{\infty}^c$ is not strictly assembled by any aTAM system.
\end{theorem}

\begin{proof}  An illustration of the following proof is provided in Figure~\ref{aTAM_basicIdea}.
Towards a contradiction, suppose $\Gamma = (T,\sigma, \tau)$ strictly assembles $S_{\infty}^c$ for some scale factor $c$.  Suppose $\sigma$ occurs at point $p_\sigma$.  Consider large enough integer $\texttt{base}$ such that $S^c_{\texttt{base}}$ contains point $p_\sigma$.  Consider a second integer $\texttt{big} >> \texttt{base}$ along with an assembly sequence $\overrightarrow{\alpha}$ of $\Gamma=(T,\sigma,\tau)$ that tiles all points in $S^c_{\texttt{big}}$.  Note that $\overrightarrow{\alpha}$ must exist for any such integer $\texttt{big}$ as $\Gamma$ strictly assembles $S_{\infty}^c$.

\begin{figure}
\begin{center}
	\includegraphics[width=.95\textwidth]{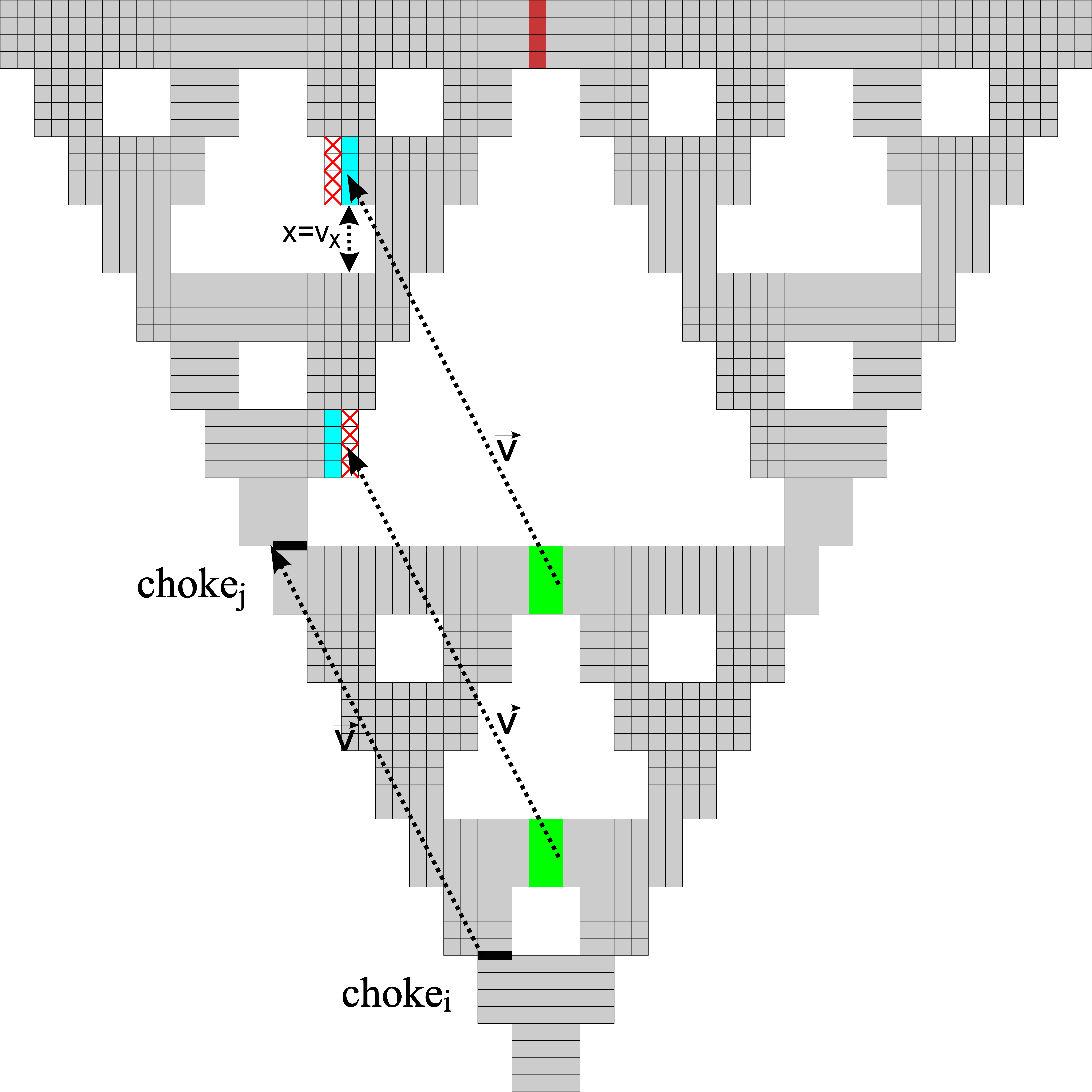}
    \caption{Simplified illustration of the proof. We consider the tiling of the $x=-1$ column of the fractal above $choke_j$ (shown in red). Tiling this column requires passing the $x=v_x$ column. As the chokes share equivalent window-movies, we can consider an alternate assembly sequence that tiles the region after $choke_i$ in the same manner as $choke_j$. If $S^c_{\infty}$ is self-assembled, it must be that the $x=-1$ and $x=0$ columns above $choke_i$ (shown in green) are placed with specific $y$-ordinates. According to Lemma~\ref{lemma:pointlanding}, when these tiles are translated back to $choke_j$ (shown in blue), either the left or right tiles lie outside of $S^c_{\infty}$ (shown by x'd tiles).}
	\label{aTAM_basicIdea}
\end{center}
\end{figure}

\begin{figure}
\begin{center}
	\includegraphics[width=.55\textwidth]{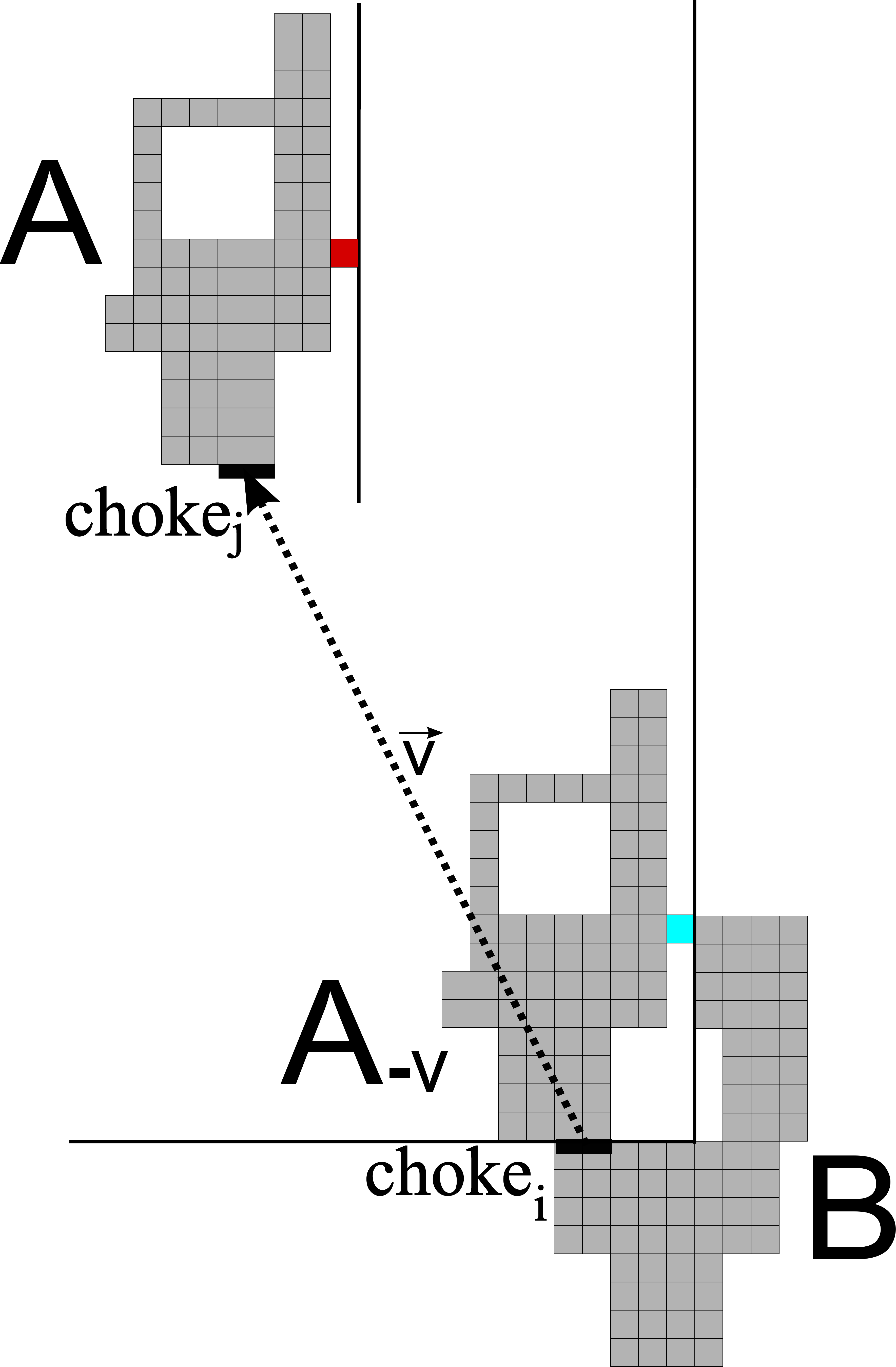}
    \caption{Illustrated here are examples of possible subassemblies $A, A_{-\vec{v}}$ and $B$. By definition, there are no tiles in $A$ east of the $x=v_x$ column. Thus, when we translate $A$ by vector $-\vec{v}$, there are no tiles east of the $x=0$ column. $B$ contains no tiles west of $x=1$ above $\texttt{choke}_i$, therefore ensuring that $A_{-\vec{v}}$ and $B$ contain no elements with the same coordinates. Thus, by also ensuring equivalent window-movies along $choke_i$ and $choke_j$, $A_{-\vec{v}} \bigcup B$ is producible. The single tile in this example of $A$ in the $x=v_x$ column, $\texttt{right}$, is shown in red, and its translated copy $\texttt{right}_{-\vec{v}}$ is shown in blue.}
	\label{aTAM_subassemblies}
\end{center}
\end{figure}

We now consider the window-movies for $\overrightarrow{\alpha}$ at the choke points that lie between $S^c_{\texttt{base}}$ and $S^c_{\texttt{big}}$.  Define $\texttt{choke}_i$ as follows:  Consider the first point tiled by $\overrightarrow{\alpha}$ with $x$-ordinate value of -1 or 0 (the center two columns of the fractal) that lie north of the edges of $\texttt{choke}^{\ell}_i$ and $\texttt{choke}^\texttt{r}_i$.  If this first tiled center point has $x$-ordinate value -1 (i.e., the tile is placed ``from the left"), define $\texttt{choke}_i \triangleq \texttt{choke}^{\ell}_i$, and define $\texttt{side}^{\overrightarrow{\alpha}}_i \triangleq 0$. Define $\texttt{choke}_i \triangleq \texttt{choke}^{r}_i$ symmetrically if the first tiled center column point has $x$-ordinate value $0$ (i.e., the tile is placed ``from the right"), and define $\texttt{side}^{\overrightarrow{\alpha}}_i\triangleq 1$.

Now assume $\texttt{big}$ is large enough such that there exist $i,j$, $\texttt{big} > j > i > \texttt{base}$, such that 1) $\texttt{side}_i = \texttt{side}_j$, and 2) $M_{\overrightarrow{\alpha},\texttt{choke}_i}$ is equivalent to $M_{\overrightarrow{\alpha},\texttt{choke}_j}$.  Note that such a $\texttt{big}$ must exist by the pigeon-hole principle since there are an infinite number of chokes in $S_{\infty}^c$.

We now construct a pair of points that must be placed by $\overrightarrow{\alpha}$.  Further, we utilize the established equivalent window-movies and Lemma~\ref{lemma:pointlanding} to show that one of these two points must lie outside of $S^c_{\infty}$, thus implying that $\overrightarrow{\alpha}$ does not self-assemble $S^c_{\infty}$.  The pair of tiles placed at these locations will be called \texttt{left} and \texttt{right}.


Without loss of generality, assume $\texttt{side}_i = \texttt{side}_j = 0$.  Let $\vec{v}=(v_x, v_y)$ denote the translation difference from $\texttt{choke}_i$ to $\texttt{choke}_j$.  Consider the subassembly of $\overrightarrow{\alpha}$'s final produced assembly consisting of the tiled points north of the edge $\texttt{choke}_j$ and west of the $x=0$ column that are elements of the largest assembly in $\overrightarrow{\alpha}$ which contains exactly one tile in the $x=v_x$ column north of $\texttt{choke}_j$.  Call this subassembly $A$, call the tile north of $\texttt{choke}_j$ in the $x=v_x$ column $\texttt{right}$, and call the tile placed west of this tile $\texttt{left}$. Note that the translations of the $\texttt{right}$ and $\texttt{left}$ tiles by vector $-\vec{v}$, $\texttt{right}_{-\vec{v}}$ and $\texttt{left}_{-\vec{v}}$, lie in the center two respective columns $x=-1$ and $x=0$ of the fractal.  Let $A_{-\vec{v}}$ denote the translation of $A$ by vector $-\vec{v}$. Examples of $A$ and $A_{-\vec{v}}$ are shown in Figure~\ref{aTAM_subassemblies}.

Now consider the subassembly $B$ defined to be the subassembly consisting of any tiles below $\texttt{choke}_i$ or east of column $x=0$ that are elements of the largest assembly in $\overrightarrow{\alpha}$ that does not contain any tiles in the $x=0$ column north of $\texttt{choke}_i$. An example of $B$ shown in Figure~\ref{aTAM_subassemblies}.  Now observe that $A_{-\vec{v}} \bigcup B$ is producible.  This is due to the equivalent window-movies that produced the original untranslated $A$, along with the constraint that $B$ does not cross the middle columns of the fractal, which allows $A_{-\vec{v}}$ to grow in isolation in the same fashion $A$ was assembled.  As $A_{-\vec{v}} \bigcup B$ is producible by a system that allegedly assembles $S^c_{\infty}$, all tiles in $A_{-\vec{v}} \bigcup B$ within the center $x=-1$ and $x=0$ columns, which includes $\texttt{right}_{-\vec{v}}$ and $\texttt{left}_{-\vec{v}}$, must occur at $y$ location between $c2^{k+1}$ and $c2^{k+1} + 2c-1$ for a positive integer $k > i$. Note that if $\texttt{right}_{-\vec{v}}$ and $\texttt{left}_{-\vec{v}}$ fall outside of this region, they are inconsistent with $S^c_{\infty}$. Thus, the location of $\texttt{right}_{-\vec{v}}$ and $\texttt{left}_{-\vec{v}}$ satisfies the constraints for points $p^\ell$ and $p^r$ from Lemma~\ref{lemma:pointlanding}, which in turn tells us that $\texttt{right}$ and $\texttt{left}$, which are both placed by a valid assembly sequence, are not both contained in $S^c_{\infty}$.  Therefore, $\Gamma$ does not self-assemble $S^c_{\infty}$.
\end{proof}

%% file: near_perfect_impossibility.tex
\subsection{Perfection Requires Many Hands}
In this section we show that achieving near-perfect assembly of the Sierpinski triangle and the Sierpinski carpet requires at least 3 and 8 hands respectively.

\begin{theorem}\label{thm:nearPerfectImpossible}
Consider a multiple hand system $(T,\tau,h)$.
\begin{itemize}
    \item If $h < 3$, then $\Gamma$ does not near-perfectly assemble $S^c_\infty$ for any scale factor $c$.
    \item If $h < 8$, then $\Gamma$ does not near-pefectly assemble $C^c_\infty$ for any scale factor $c$.
\end{itemize}
\end{theorem}
\begin{proof}
We prove the result in the case of the Sierpinski triangle.  The result for the Sierpinski carpet is realized in a similar fashion.  Towards a contradiction, suppose $\Gamma=(T,\tau,2)$ near-perfectly assembles $S^c$ for some positive integer $c$.  Let constant $d$ be a constant at least as large as the constants $\Gamma$ uses to satisfy constraints (1) and (2) of the near-perfect self-assembly definition.  Consider $S_i$ such that $|S_i| > d$.  By constraint (1) of near-perfection, there exists an assembly $A \in \texttt{PROD}_\Gamma$ whose size is such that $|S_i| \geq |A| \geq |S_i| - d$.  As $h=2$, we know there exists $B_1, B_2 \in \texttt{PROD}_\Gamma$ such that $\{B_1,B_2\}$ is $\tau$-combinable into $A$.  Further, we know that there must exist such an $A$ such that $|B_1| < |S_i|-d$ and $|B_2| < |S_i|-d$ (if this weren't the case, then grab the larger $B$ assembly and try again, repeating until the constraint is satisfied).  Finally, as $B_1$ and $B_2$ combine to make $A$, we know that $B'\triangleq \max(|B_1|,|B_2|) > 1/2|A| \geq 1/2|S_i| -1/2d$.  By considering large enough $i$ with respect to constant $d$, we further get that $|B'| > 1/2|S_i|-1/2d >1/3|S_i|\geq |S_{i-1}|$, yielding $|S_i|-d > |B'| > |S_{i-1}|$.  As $B\in \texttt{PROD}_\Gamma$, this contradicts property (2) of near-perfect assembly as $B'$ is too small to be close to shape $S_{i}$, and too big to be a subset of shape $S_{i-1}$.  With regard to the Sierpinski carpet, this contradiction is drawn in a similar fashion as above for all $2\leq h\leq7$.
\end{proof}

%% file: future_work.tex
\section{Future Work}
\label{sec:futureWorks}

There are a number further research directions related to this work.  First, knowledge of strictly assembled fractals within the aTAM is still incomplete.  Impossibility results have shown that a large class of ``pinch-point" fractals~\cite{jSADSSF}, ``tree" fractals~\cite{STFDNSS}, and most generally ``pier" fractals~\cite{SPFDNSS} cannot be self-assembled in the aTAM at any scale.  Extending these impossibility proofs to essentially all non-trivial fractal shapes remains open.  In particular, the Sierpinski carpet does not fall into these classes but is conjectured to not be self-assembled in the aTAM.  When combined with our positive $2$-HAM result, this proof would provide an interesting gap in power between the aTAM and the $2$-HAM for a natural and well studied shape.

The strict self-assembly of fractals in the $2$-HAM is still not well understood.  Almost none of the impossibility results for the aTAM hold within the $2$-HAM.  As an example, while it is known that the scale 1 tree Sierpinski triangle cannot be self-assembled in the 2-HAM, it is only \emph{conjectured} to be impossible at a constant scale.  Beyond the 2-HAM, we can consider the \emph{hand complexity} of an infinite shape $S$, which is the smallest number of hands needed for an $h$-HAM system to self-assemble $S$.  Is the hand complexity of the non-tree Sierpinski triangle 3?  Moreover, we see that in the case of the non-tree Sierpinski triangle, an increase in hands from 3 to 6 allows for a much simpler construction with respect to tile complexity. Is there a more general connection between tile complexity, or other efficiency metrics, and hand complexity?

The consideration of multiple handed self-assembly introduces a number of additional directions beyond fractal assembly.  For example, what finite shapes can be built more efficiently with more hands?  Can a shape built with a given number of hands be transformed into an equivalent system with fewer hands, perhaps at the cost of scale factor?  How are well studied computational complexity questions such as \emph{producibility verification} and \emph{unique assembly verification}~\cite{2HABTO} affected by the $h$ parameter in an $h$-HAM systems?  Another interesting direction is to consider $h$-HAM assembly as an error model in which an intended 2-HAM system must be designed be robust against assembly growth with up to $h$ hands.  Can general 2-HAM systems be made robust to such constraints at a reasonable cost to tile complexity and scale factor?